\documentclass[10pt,times]{llncs}

\usepackage{graphicx,color,rotating}
\usepackage{times}
\usepackage{geometry}
\usepackage{url}
\usepackage{amsmath}
\usepackage{amssymb}
\usepackage{algorithm}
\usepackage{algorithmic}
\usepackage{graphicx}
\usepackage{subfig}
\usepackage{times}
\usepackage[T1]{fontenc}

\algsetup{indent=3em}

\begin{document}

\author{Sergio Castillo-Perez$^{1}$ \and Joaquin Garcia-Alfaro$^{2}$}

\title{On the Use of Latency Graphs for\\ the Construction of Tor Circuits}

\institute{Department of Information and Communications Engineering
    (dEIC), \\
    Autonomous University of Barcelona, 08193 Bellaterra, Spain\\
{\tt Sergio.Castillo@uab.es}\\
\bigskip
  \and
Institut MINES-TELECOM, TELECOM Bretagne, \\
    35576 C\'esson-S\'evign\'e, France.\\
{\tt joaquin.garcia-alfaro@acm.org}\\
}

\maketitle

\begin{abstract}
The use of anonymity-based infrastructures and anonymisers is a
plausible solution to mitigate privacy problems on the Internet. Tor
(short for {\em The onion router}) is a popular low-latency anonymity
system that can be installed as an end-user application on a wide
range of operating systems to redirect the traffic through a series of
anonymising proxy circuits. The construction of these circuits
determines both the latency and the anonymity degree of the Tor
anonymity system. While some circuit construction strategies lead to
delays which are tolerated for activities like Web browsing, they can
make the system vulnerable to linking attacks. We evaluate in this
paper three classical strategies for the construction of Tor circuits,
with respect to their de-anonymisation risk and latency performance.
We then develop a new circuit selection algorithm that considerably
reduces the success probability of linking attacks while keeping a
good degree of performance. We finally conduct experiments on a
real-world Tor deployment over PlanetLab. Our experimental results
confirm the validity of our strategy and its performance increase for
Web browsing.\\

\textbf{Keywords}: Anonymity, Privacy, Entropy, Graphs, Algorithmics.
\end{abstract}


\section{Introduction}
\label{sec:intro}
Several anonymity designs have been proposed in the literature with
the objective of achieving anonymity on different network
technologies. From simple pseudonyms \cite{troncoso2011} to complex
unstructured protocols \cite{bansod08}, anonymity solutions can offer
either strong anonymity with high latency (useful for high latency
services, such as email and usenet messages) or weak anonymity with
low-latency (useful, for instance, for Web browsing). The most
widely-used low-latency solution for traditional Internet
communications is based on anonymous mixes and onion routing
\cite{Syverson1997}. It is distributed as a free software
implementation known as Tor (\textit{The onion router}
\cite{Dingledine2004}). It can be installed as an end-user application
on a wide range of operating systems to redirect the traffic of
low-latency services with a very acceptable overhead. Tor's objective
is the protection of privacy of a sender as well as the contents of
its messages. To do so, it transforms cryptographically those messages
and mixes them via a circuit of routers. The circuit routes the
message in an unpredictable way. The content of each message is
encrypted for every router in the circuit, with the objective of
achieving anonymous communication even if a set of routers are
compromised by an adversary. Upon reception, a router decrypts the
message using its private key to obtain the following hop and
cryptographic material on the path. This path is initially defined at
the beginning of the process. Only the entity that creates the circuit
knows the complete path to deliver a given message. The last router
of the path, the \textit{exit} node, decrypts the last layer and delivers
an unencrypted version of the message to its target.

Tor allows the construction of anonymous channels with latency enough
to route traffic for services like the Web \cite{Kopsell2006}.
However, it might still impact its performance depending on the
specific strategy used for the establishment of the channel. In this
paper, we address the influence of circuit construction strategies on
the anonymity degree of Tor. We first provide a formal definition of
the selection of Tor nodes process, of the adversary model targeting
the communication anonymity of Tor users, and an analytical expression
to compute the anonymity degree of the Tor infrastructure based on the
circuit construction criteria. Based on these definitions, we evaluate
three classical strategies, with respect to their de-anonymisation
risk, and regarding their performance for anonymising Internet
traffic. We then present the construction of a new circuit selection
algorithm that aims at reducing the success probability of linking
attacks while providing enough performance for low-latency services. A
series of experiments, conducted on a real-world Tor deployment over
PlanetLab \cite{leon2011} confirm the validity of the new strategy,
and show its superiority over the classical ones.\\

\noindent \textbf{Paper organisation ---} Section \ref{sec:rationale}
presents the rationale of our work. Section \ref{sec:strategies}
evaluates the anonymity degree of three traditional strategies for the
construction of Tor circuits. Section \ref{sec:new-strategy} presents
our new strategy. Section \ref{sec:new-strategy-evaluation} evaluates
the anonymity degree of our solution. Section \ref{sec:experiments}
experimentally evaluates the latency performance of each strategy
using PlanetLab. Section \ref{sec:sota} surveys related work. Section
\ref{sec:conclusion} concludes the paper.

\section{Rationale}
\label{sec:rationale}

In this section, we introduce the notation, models, and core definitions that
are necessary to understand the rationale of our work.

\subsection{Tor circuit}

\label{sec:notation}
Formally, we can describe a connection using the Tor network as follows. First,
we define a client node $s$ called a \textit{client} or \textit{onion proxy},
and a \textit{destination server} node $d$ which we want to interconnect to
exchange data in an anonymous manner. Let $N$ be the set of nodes deployed in
the Tor network, and $n=|N|$ the cardinality of the set. Let node $e \in N$
denote a specified node, called the \textit{entrance node}, and $x \in N$ the
\textit{exit node}. Then, a \textit{Tor circuit} is a sequence of nodes
$C=\langle s,e,r_{1},r_{2}, ... ,r_{l},x \rangle$, where $r_{i} \in N$ is any
\textit{intermediary node}. The nodes $e$, $x$, and $r_i$, $i \in \{1, ...,
l\}$, are also known as \textit{onion routers}. We define the \textit{path of a
circuit} as the set of links (i.e., network connections) $P=\{a_{1}, ...,
a_{l+2}\}$ associated to the \textit{Tor circuit}, where $a_{1}=(s, e)$,
$a_{2}=(e, r_{1})$, $a_{3}=(r_{1}, r_{2})$,~...~, $a_{l+1}=(r_{l-1},r_{l})$,
$a_{l+2}=(r_l, x)$. The value $|P|=l+2$ is called the \textit{length of the
circuit}. A \textit{connection using the Tor network} is composed by the client
and destination nodes interconnected through a Tor circuit as follows:

\begin{center}
$\underbrace{s \xrightarrow{a_1} e \xrightarrow{a_2} r_1
    \xrightarrow{a_3} r_2 \xrightarrow{a_4}... \xrightarrow{a_{l}}
    r_{l-1} \xrightarrow{a_{l+1}} r_l \xrightarrow{a_{l+2}}
    x}_{\text{Tor network}} \rightarrow d$
\end{center}

\subsection{Adversary model}
\label{sec:adv_model}
The adversary assumed in our work relies on the threat mo\-del proposed by
Syverson \textit{et al.} in \cite{Syverson2001}. Such a pragmatic model
considers that, regardless of the number of onion routers in a circuit, an
adversary controlling the entrance and exit nodes would have enough information
in order to compromise the communication anonymity of a Tor client. Indeed, when
both nodes collude, and given that the entry node knows the source of the
circuit, and the exit node knows the destination, they can use traffic analysis
to link communication over the same circuit \cite{hopper2010}.


Assuming the model proposed in \cite{Syverson2001}, then an adversary who
controls $c > 1$ nodes over the $n$ nodes in the Tor network can control an
entry node with probability $(\frac{c}{n})$, and an exit node with probability
$(\frac{c}{n})$. This way, the adversary may de-anonymise the traffic flowing on
a controlled circuit (i.e., a circuit whose entry and exit nodes are controlled
by the adversary) with probability $(\frac{c}{n})^2$ if the length of the
circuit is greater than two; or $\frac{c(c-1)}{n^2}$ if the length of the
circuit is equal to two (cf. \cite{Syverson2001} and citations thereof).
Adversaries can determine when the nodes under their control are either entry or
exit nodes for the same circuit stream by using attacks such timing-based
attacks \cite{Back2001}, fingerprinting \cite{Murdoch2006}, and several other
existing attacks.

Let us observe that the aforementioned probability of success assumes
that the probability of a node from being selected on a Tor circuit is
randomly uniform, that is, the boundaries provided in
\cite{Syverson2001} only apply to the standard (random) selection of
nodes, hereinafter denoted as \textit{random selection of nodes
  strategy}. Given that the goal of our paper is to evaluate
alternative selection strategies, we shall adapt the model. Therefore,
let $p_1$, $p_2$, $p_3$, $\ldots$, $p_c$ be the corresponding
selection probabilities assigned by the circuit construction algorithm
to each node controlled by the adversary, then the probability of
success corresponds to the following expression:
\begin{equation*}
(p_1 + p_2 + p_3 + \ldots + p_c) \cdot (p_1 + p_2 + p_3 + \ldots + p_c)
\end{equation*}
that can be simplified as:
\begin{equation*}
\Big ( \sum_{i=1}^{c}p_i \Big )^2
\end{equation*}
Following is the analysis.

\begin{theorem}
\label{theo:probattack}
Let $c$ be the number of nodes controlled by the adversary. Let the
Tor client use a selection criteria which, for a certain circuit,
every node selection is independent. Let $p_1$, $p_2$, $p_3$,
$\ldots$, $p_c$ be the corresponding selection probabilities assigned
by the circuit construction algorithm to each node controlled by the
adversary. Then, the success of the adversary to compromise the
security of the circuit is bounded by the following probability:\\
\begin{equation*}
\Big ( \sum_{i=1}^{c}p_i \Big )^2
\end{equation*}
\end{theorem}

\medskip

\begin{proof}
The proof is direct by using the sum and product rules of probability theory,
and taking into account that the selection of every node is an independent
event. First, the probability of selecting the entrance or exit node in the set
of nodes controlled by the adversary is (sum rule):
\begin{equation*}
\sum_{i=1}^{c}p_i
\end{equation*}
Then, the probability of selecting, at the same time, a controlled entrance and
exit node in a circuit is (product rule):
\begin{equation*}
\Big (\sum_{i=1}^{c}p_i \Big ) \Big (\sum_{i=1}^{c}p_i \Big )=\Big (
\sum_{i=1}^{c}p_i \Big )^2
\end{equation*}
\end{proof}

\begin{corollary}
\label{coro:syversoncastillo}
The Syverson \textit{et al.} success probability boundary in
\cite{Syverson2001}, i.e., $(\frac{c}{n})^2$, is equivalent to the boundary
defined in Theorem~\ref{theo:probattack} when the circuit selection criteria is
a random selection of nodes.
\end{corollary}

\begin{proof}
Let $N$ be the set of nodes deployed in a Tor network with $n=|N|$,
and let $A \subseteq N$ be the subset of nodes controlled by an
adversary with $c=|A|$. The probability of a node $n_i \in N$ to be
selected is $p_i=\frac{1}{n}$. Then, by applying it to the boundary
defined in Theorem~\ref{theo:probattack}, we obtain:
\begin{equation*}
\Big (\sum_{i=1}^{c}p_i \Big )^2=(c \cdot p_i)^2=\Big (c \frac{1}{n}
\Big)^2= \Big (\frac{c}{n} \Big)^2
\end{equation*}
\end{proof}

\subsection{Anonymity degree}

Most work in the related literature has used the (Shannon) entropy
concept to measure the anonymity degree of anony\-misers like Tor (cf.
\cite{Diaz2002,Danezis2002} and citations thereof). We recall that the
entropy is a measure of the uncertainty associated with a random
variable, that can efficiently be adapted to address new networking
research problems \cite{tomlin00,fan09,tellenbach11}. In this paper,
the entropy concept is used to determine how predictable is the
selection of the nodes in accordance to a given strategy or, in other
words, how easy is to violate the anonymity in relation to the
adversary model defined in Section~\ref{sec:adv_model}. Formally,
given a probability space $(\Omega, \mathcal{F}, \mathbb{P})$ with a
sample space $\Omega=\{\omega_1,\omega_2, ..., \omega_n\}$ where
$\omega_i$ denotes the outcome of the node $n_i \in N$ ($\forall i \in
\{1, ..., n\}$), a $\sigma$-field $\mathcal{F}$ of subsets of
$\Omega$, and a probability measure $\mathbb{P}$ on $(\Omega,
\mathcal{F})$, we consider a random discrete variable $X$ defined as
$X:\Omega \rightarrow \mathbb{R}$ that takes values in the countable
set $\{x_1, x_2, ..., x_n\}$, where every value $x_i\in \mathbb{R}$
corresponds to the node $n_i \in N$. The discrete random variable X
has a $\mathrm{pmf}$ (probability mass function) $f: \mathbb{R}
\rightarrow [0,1]$ given by $f(x_i) = p_i = \mathbb{P}(X=x_i)$. Then,
we define the entropy of a discrete random variable (i.e., the entropy
of a Tor network) as:
\begin{equation}
H(X) = - \displaystyle \sum_{i=1}^{n}p_i\cdot log_{2}(p_i)
\label{eq:entropy}
\end{equation}
Since the entropy is a function whose image depends on the number of nodes, with
property $H(X) \ge 0$, it cannot be used to compare the level of anonymity of
different systems. A way to avoid this problem is as follows. Let $H_M(X)$ be
the maximal entropy of a system, then the entropy that the adversary may obtain
after the observation of the system is characterised by $H_M(X) - H(X)$. The
maximal entropy $H_M(X)$ of the network applies when there is a uniform
distribution of probabilities (i.e., $\mathbb{P}(X=x_i)=p_i= \frac{1}{n}$,
$\forall i\in \{1,...,n\}$), and this leads to $H(X)=H_M(X)=log_2(n)$. The
anonymity degree shall be then be defined as:\\
\begin{equation}
d = 1 - \frac{H_M(X)-H(X)}{H_M(X)} = \frac{H(X)}{H_M(X)}
\label{eq:anondegree}
\end{equation}
Note that by dividing $H_M(X) - H(X)$ by $H_M(X)$, the resulting expression is
normalised. Therefore, it follows immediately that $0 \leq d \leq 1$.

\subsection{Selection criteria}

Taking into account the aforementioned anonymity degree expression, we can now
formally define a selection of Tor nodes criteria as follows.

\begin{definition} A selection of Tor nodes criteria is an algorithm executed by
a Tor client $s$ that, from a set of nodes $N$ with $n=|N|$ and a length of a
circuit $\delta$, selects ---using a given policy--- the entrance node $e$,
the exit node $x$, and the intermediary nodes $r_i$, $\forall i
\in\{1,...,\delta-2\}$, and outputs its corresponding circuit $C=\langle
s,e,r_{1},r_{2}, ... ,r_{\delta-2},x \rangle$ with a path $P=\{a_{1}, ...,
a_{\delta}\}$, where $a_{1}=(s, e)$, $a_{2}=(e, r_{1})$, $a_{3}=(r_{1},
r_{2})$,~...~, $a_{\delta-1}=(r_{\delta-3}$, $r_{\delta-2})$,
$a_{\delta}=(r_{\delta-2}, x)$. We use the notation convention $\psi(N,\delta)$
to denote the algorithm. The policy for the selection criteria of nodes can be
modelled as a discrete random variable $X$ that has a  $\mathrm{pmf}$ $f(x)$,
and we use the notation $\psi(N,\delta)\sim f(x)$.
\end{definition}

\section{Anonymity degree of three classical circuit construction strategies}
\label{sec:strategies}

In this section, we present three existing strategies for the construction of
Tor circuits, and elaborate on the conceptual evaluation of their anonymity
degree.

\subsection{Random selection of nodes}
\label{sec:random_selection}

The random selection of Tor nodes is an algorithm $\psi_{rnd}(N, \delta) \sim
f_{rnd}(x)$ with an associated discrete random variable $X_{rnd}$. The procedure
associated to this selection criteria is outlined in Algorithm \ref{alg:random}.
The selection policy of $\psi_{rnd}(N, \delta)$ is based on uniformly choosing
at random those nodes that will be part of the resulting circuit. Thus, the
$\mathrm{pmf}$ $f_{rnd}(x)$ is defined as follows:
\begin{equation*}
f_{rnd}(x_i)=p_i=\mathbb{P}(X_{rnd}=x_i)=\frac{1}{n}
\end{equation*}

Hence, the entropy of a Tor network whose clients use a random selection of
nodes is characterised by the following expression:
\begin{equation*}
H_{rnd}(X_{rnd}) = -\sum_{i=1}^{n}\frac{1}{n}\cdot log_{2}\Bigl (\frac{1}{n}
\Bigr ) =
\end{equation*}
\begin{equation*} -\frac{1}{n}\sum_{i=1}^{n}{(log_2(1)-log_2(n))} = log_2(n)
\end{equation*}

\medskip

\begin{algorithm}[b]
\caption{Random Selection of Nodes - $\psi_{rnd}(N,\delta)$}
\label{alg:random}
\begin{algorithmic}
\STATE \textbf{Input:} $s,N,\delta$\\
\STATE \textbf{Output:} $C=\langle s,e,r_{1},r_{2}, ... ,r_{\delta-2},x \rangle$,
$P=\{a1,...,a_\delta\}$\\
\medskip
	\STATE $M \leftarrow N$
	\STATE $C \leftarrow \{s\}$
	\FOR{$i \leftarrow 1$ \TO $\delta$}
		\STATE $j \leftarrow $random$(1, |M|)$
		\STATE $C \leftarrow C \cup \{m_j~|~m_j \in M\}$
		\STATE $P \leftarrow P \cup \{(c_i, c_{i+1})\}$
		\STATE $M \leftarrow M\setminus\{m_j~|~m_j \in M\}$
	\ENDFOR
\end{algorithmic}
\end{algorithm}

\begin{theorem}
\label{theo:rnd}
The selection of Tor nodes $\psi_{rnd}(N, \delta)\sim f_{rnd}(x)$ with an
associated discrete random variable $X_{rnd}$ gives the maximum degree of
anonymity among all the possible selection algorithms.
\end{theorem}

\begin{proof}
The proof is direct by replacing $H_{rnd}(X_{rnd})$ in Equation
(\ref{eq:anondegree}):

\begin{equation*}
d_{rnd}=\frac{H_{rnd}(X_{rnd})}{H_M(X_{rnd})} =\frac{log_2(n)}{log_2(n)} = 1
\end{equation*}
\end{proof}
\medskip

\subsection{Geographical selection of nodes}
\label{sec:geoip_selection}

The geographical selection of Tor nodes is an algorithm $\psi_{geo}(N,
\delta) \sim f_{geo}(x)$ with an associated discrete random variable
$X_{geo}$. Its selection method is based on uniformly choosing the
nodes that belong to the same country of the client $s$ that executes
$\psi_{geo}(N,\delta)$. The aim of this strategy is to reduce the
latency of the communications using the Tor network, since the number
of hops between Tor nodes of the same country is normally smaller than
the number of hops between nodes that are located at different
countries. Algorithm \ref{alg:geoip} summarises the procedure
associated with this selection criteria.

\begin{algorithm}
\caption{Geographical Selection of Nodes - $\psi_{geo}(N,\delta)$}
\label{alg:geoip}
\begin{algorithmic}
\STATE \textbf{Input:} $s,N,\delta$, $K_c$\\
\STATE \textbf{Output:} $C=\langle s,e,r_{1},r_{2}, ... ,r_{\delta-2},x \rangle$,
$P=\{a1,...,a_\delta\}$\\
\medskip
	\STATE $M \leftarrow \{n_i\in N~|~g_c(n_i)=K_c\}$
	\STATE $C \leftarrow \{s\}$
	\FOR{$i \leftarrow 1$ \TO $\delta$}
		\STATE $j \leftarrow $random$(1, |M|)$
		\STATE $C \leftarrow C \cup \{m_j~|~m_j\in M\}$
		\STATE $P \leftarrow P \cup \{(c_i, c_{i+1})\}$
		\STATE $M \leftarrow M\setminus\{m_j~|~m_j \in M\}$
	\ENDFOR
\end{algorithmic}
\end{algorithm}

Formally, we define a function $g_c: \mathbb{R}\rightarrow \mathbb{N}$
that, given a certain node $x_i \in X_{geo}$, returns a number that
identifies its country. Thus, given the specific country number $K_c$
of the client node $s$, the $\mathrm{pmf}$ $f_{geo}(x)$ is
characterised by the following expression:
\begin{equation*}
f_{geo}(x_i)=p_i=\mathbb{P}(X=x_i)=\begin{cases}
\frac{1}{m}, &\text{if $g_c(x_i)=K_c$;}\\
0             &\text{otherwise.}
\end{cases}
\end{equation*}
where $m = |\{x_i \in X_{geo}~|~g_c(x_i)=K_c\}|$. Then, the entropy of a system
whose client nodes use a geographical selection for a certain country $K_c$ is:
\begin{equation*}
H_{geo}(X_{geo}) = - \displaystyle \sum_{i=1}^{m}\frac{1}{m}\cdot log_2
\Bigr(\frac{1}{m}\Bigl)=log_2(m)
\end{equation*}
Therefore, by replacing the previous expression in Equation
(\ref{eq:anondegree}), the anonymity degree is equal to:
\begin{equation*}
d_{geo} = \frac{log_2(m)}{log_2(n)}
\end{equation*}

\begin{theorem}
\label{theo:bw1}
The maximum anonymity degree of a Tor network whose clients use a geographical
selection of nodes is achieved iff all the nodes are in the same fixed country
$K_c$.
\end{theorem}

\begin{proof}
($\Rightarrow$) Given $d_{geo}=\frac{log_2(m)}{log_2(n)}$ for the country $K_c$
of a particular client $s$, we can impose the restriction of maximum degree of
anonymity:
\begin{equation*}
d_{geo} = \frac{log_2(m)}{log_2(n)} = 1
\end{equation*}
Hence,
\begin{align*}
log_2(m) &= log_2(n)\\
2^{log_2(m)} &= 2^{log_2(n)}\\
m &= n
\end{align*}
($\Leftarrow$) If $g_c(x_i)=K_c$, $\forall x_i \in X_{geo}$, then we have that
$m = |\{x_i \in X_{geo}~|~g_c(x_i)=K_c\}| = |N|$. Thus,

\begin{equation*}
d_{geo} = \frac{log_2(m)}{log_2(n)} = \frac{log_2(n)}{log_2(n)} = 1
\end{equation*}
\end{proof}

\begin{theorem}
\label{theo:bw2}
Given a Tor network whose clients use the algorithm $\psi_{geo}(N,
\delta) \sim f_{geo}(x)$ for a fixed country $K_c$, and with an
associated discrete random variable $X_{geo}$, the anonymity degree is
increased as $m$ approaches $n$ (i.e., $m \to n$), where $m = |\{x_i
\in X_{geo}~|~g_c(x_i)=K_c\}|$ and $n=|N|$.
\end{theorem}

\begin{proof}
It suffices to prove that $d_{geo}$ is a monotonically increasing
function. That is, we must prove that $\frac{\partial}{\partial m}
(d_{geo})>0$, $\forall m>0$. Therefore, the proof is direct by
deriving, since the inequality:
\begin{equation*}
\frac{\partial}{\partial m}\Bigl(\frac{log_2(m)}{log_2(n)}\Bigr)=
\frac{1}{m\cdot log(n)} > 0
\end{equation*}
is true $\forall m>0$ and $\forall n>1$. We must notice that, from the point of
view of a Tor network, the restriction of the number of nodes $n>1$ makes sense,
since a network with $n \leq 1$ nodes becomes useless as a way to provide an
anonymous infrastructure.
\end{proof}

Figure \ref{fig:geo} depicts the influence of the uniformity of the
number of nodes per country on the anonymity degree. It shows, for a
fixed country, the anonymity degree of four Tor networks in function
of the nodes that are located in that country with respect to the
total number of nodes of the network. The considered Tor networks
have, respectively, 10, 50, 100 and 200 nodes. Their anonymity degrees
are denoted as $d_{10}$, $d_{50}$, $d_{100}$ and $d_{200}$. We can
observe that the anonymity degree increases as the total number of
nodes of the same country grows up (cf. Theorem \ref{theo:bw2}). This
fact can be extended until the maximum value of anonymity is achieved,
which occurs when the number of nodes of the particular country is the
same as the nodes that compose the entire network (cf. Theorem
\ref{theo:bw1}).

\begin{figure}[hptb]
    \centering
    \includegraphics[width=9.5cm]{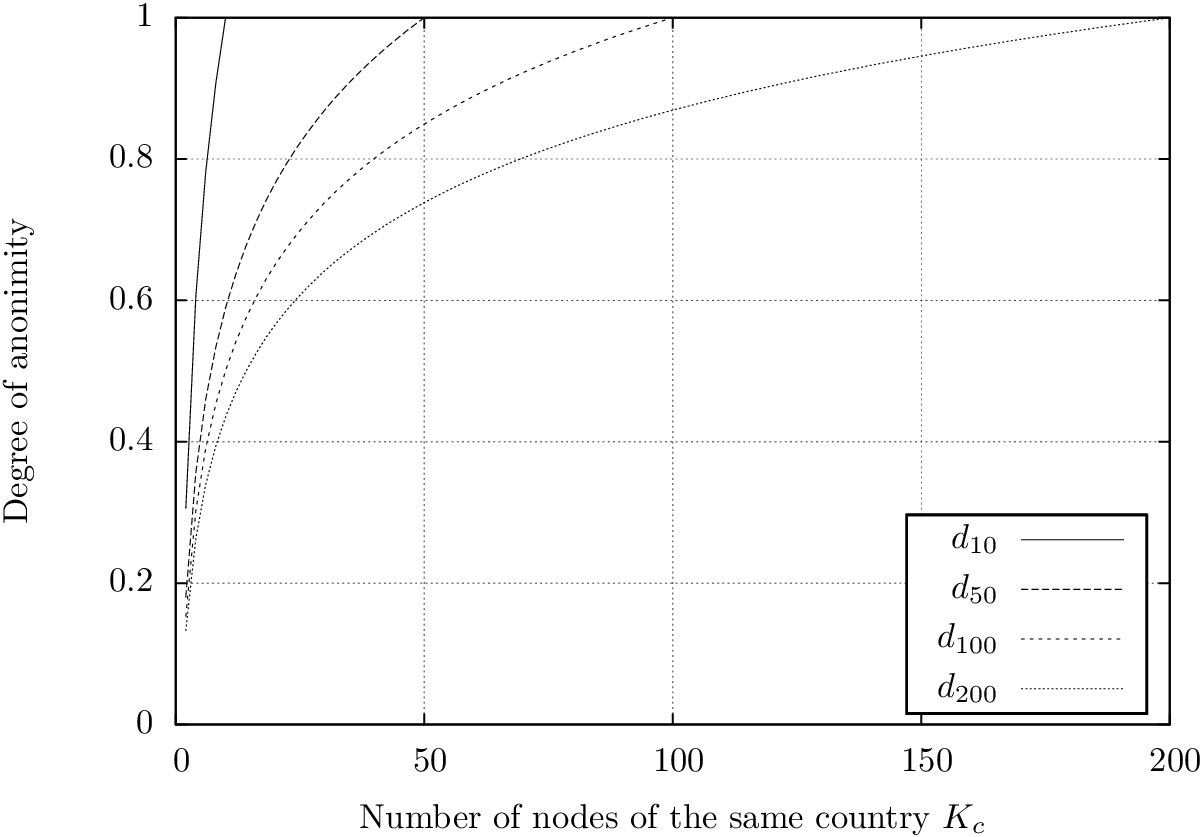}
    \caption{Influence of the uniformity of the number of nodes per country in
	     the anonymity degree for $\psi_{geo}(N,\delta)$}
    \label{fig:geo}
\end{figure}

\begin{theorem}
Given a client $s$ that uses as selection algorithm
$\psi_{geo}(N,\delta)$ in a Tor network with $n=|N|$, such that the
network nodes belong to a $p \ll n$ different countries, where $p$
is the number of different countries in Tor network, then the best
distribution of nodes that maximises the anonymity degree of the whole
system is achieved iff every country has $t=\lfloor \frac{n}{p}\rceil$
nodes.
\end{theorem}

\begin{proof}
($\Rightarrow$) Let $p$ be the number of different countries of a Tor
  network, we can consider a collection of subsets $S_1, S_2, ..., S_p
  \subseteq N$ such as $\bigcup_{i=1}^{p}S_i = N$ and
  $\bigcap_{i=1}^{p}S_i = \varnothing$. Let $t_i$ be the number of
  nodes associated to the subset $S_i$, $i \in \{1,...,p\}$. Then, the
  anonymity degree of the whole system is maximised when the sum of
  all the degrees of anonymity of every country equals 1:
\begin{align*}
\sum_{i=1}^{p} \frac{log_2(t_i)}{log_2(n)} &= 1\\
\frac{log_2(t_1)}{log_2(n)} + \frac{log_2(t_2)}{log_2(n)} + ... +
\frac{log_2(t_p)}{log_2(n)} &= 1\\
2^{log_2(t_1)} + 2^{log_2(t_2)} + ... + 2^{log_2(t_p)} &= 2^{log_2(n)}\\
t_1 + t_2 + ... + t_p &= n
\end{align*}
However, to maximise the anonymity degree of the whole system implies also to
have the same uncertainty inside every subset $S_i$, $i \in \{1,...,p\}$, or, in
other words, to have the same number of nodes in every subset. Hence, we have
$t_1 = t_2 = ... = t_p = t$ and this leads to:
\begin{align*}
t_1 + t_2 + ... + t_p &= n\\
\underbrace{t + t + ... + t}_{p~\textrm{times}} &= n\\
p \cdot t &= n\\
t &= \frac{n}{p}
\end{align*}
($\Leftarrow$) Given $t=\lfloor \frac{n}{p}\rceil$ be the number of nodes of a
certain subset $S_i$, $i \in \{1,...,p\}$, we have $\sum_{i=1}^{p} |S_i| = {p
\cdot t}=n$. The  $\mathrm{pmf}$ associated to $\psi_{geo}(N,\delta)$ is then
$f_{geo}(x)=\frac{1}{t}$ for each subset $S_i$, $i\in \{1,...,p\}$. Therefore,
the entropy of each subset (i.e., country) is:
\begin{equation*}
H_{geo}(X_{geo}) = -\sum_{i=1}^{t}\frac{1}{t}\cdot log_{2}
\Bigl(\frac{1}{t}\Bigr)= log_2(t)
\end{equation*}
Hence, for each subset $S_i$, $i \in \{1,...,p\}$, the anonymity degree can
be expressed as follows:
\begin{equation*}
d_{geo} = \frac{log_2(t)}{log_2(n)}
\end{equation*}
Suppose now, by contradiction, that there exists a unique $S_{q} \in
\{S_1, S_2, ..., S_p\}$ for a particular country $K_q$ such that $|S_q|\ne
t$, and its anonymity degree is expressed by
$d_{geo^*}=\frac{log_2(|S_q|)}{log_2(n)}$. Then, taking into account
that $d_{geo}$ and $d_{geo^*}$ are monotonically increasing functions
(cf. proof of Theorem \ref{theo:bw2}), we have two options:
\begin{itemize}
   \item If $|S_q|<t \rightarrow d_{geo^*} < d_{geo}$
   \item If $|S_q|>t \rightarrow d_{geo^*} > d_{geo}$
\end{itemize}
But this is not possible since:
\begin{align*}
   \sum_{i=1}^{p} |S_i| &= n\\
         (p-1)t + |S_q| &= n\\
                  |S_q| &= n - t(p-1)\\
                  |S_q| &= n - \frac{n}{p} (p-1)\\
                  |S_q| &= \frac{n}{p}
\end{align*}
which implies that $d_{geo^*} = d_{geo}$, contradicting the above two options.
\end{proof}

\subsection{Bandwidth selection of nodes}
\label{sec:bw_selection}

The bandwidth selection of nodes strategy is an algorithm
$\psi_{bw}(N, \delta) \sim f_{bw}(x)$ with an associated discrete
random variable $X_{bw}$ whose selection policy is based on choosing,
with high probability, the nodes with best network bandwidth. The
procedure associated to this selection criteria is outlined in
Algorithm \ref{alg:bw}. The aim of this strategy is to reduce the
latency of the communications through a Tor circuit, specially when
the communications imply a great rate of data exchanges. At the same
time, this mechanism provides a balanced anonymity degree, since the
selection of nodes is not fully deterministic from the adversary point
of view.

\begin{algorithm}
\caption{Bandwidth Selection of Nodes - $\psi_{bw}(N,\delta)$}
\label{alg:bw}
\begin{algorithmic}
\STATE \textbf{Input:} $s,N,\delta$\\
\STATE \textbf{Output:} $C=\langle s,e,r_{1},r_{2}, ... ,r_{\delta-2},x \rangle$,
$P=\{a1,...,a_\delta\}$\\
\medskip
	\emph{/* Compute a weighted well-ordered set */}
	\STATE $M \leftarrow \{n_i\in N~|~g_{bw}(n_i) \le g_{bw}(n_{i+1})\}$
	\STATE $T_{bw}  \leftarrow  \sum\limits^{n}_{i=1} g_{bw}(m_i)$, $\forall m_i \in
               M$
	\STATE $W \leftarrow  \{\}$
	\FOR{$i \leftarrow 1$ \TO $n$}
		\STATE $W \leftarrow W \cup \Bigl\{\Bigl(m_i, \sum \limits^{i}_{j=1}
		\frac{g_{bw}(m_j)}{T_{bw}}\Bigr)~|~\forall m_i,m_j \in M\Bigr\}$
	\ENDFOR
	\medskip

	\emph{/* Compute the nodes of the circuit $C$ */}
	\STATE $C \leftarrow \{s\}$
	\FOR{$i \leftarrow 1$ \TO $\delta$}
		\STATE $rnd \leftarrow $random$(0, 1)$
		\STATE Select a tuple $(m_j, bw_j) \in W$, where \\
~~~~~~~~$m_j \notin C$ \AND\\
~~~~~~~~$rnd \in [bw_j, bw_{j+1})$
		\STATE $C \leftarrow C \cup \{m_j\}$
		\STATE $P \leftarrow P \cup \{(c_i, c_{i+1})\}$
	\ENDFOR
\end{algorithmic}
\end{algorithm}

In this strategy, the entropy and the anonymity degree can be
described formally as follows. First, we define a bandwidth function
$g_{bw}: \mathbb{R} \rightarrow \mathbb{N}$ that, given a certain node
$x_i \in X_{bw}$, returns its associated bandwidth. Then, the
$\mathrm{pmf}$ $f_{bw}(x)$ is defined by the expression:
\begin{equation*}
f_{bw}(x_i)=p_i=\mathbb{P}(X_{bw}=x_i)=\frac{g_{bw}(x_i)}{T_{bw}}
\end{equation*}
where $T_{bw}={\sum\limits^{n}_{i=1} g_{bw}(x_i)}$ is the total bandwidth of the
Tor network. Hence, the entropy of a system whose clients use a bandwidth
selection of nodes strategy is:
\begin{equation*}
H_{bw}(X) = - \displaystyle \sum_{i=1}^{n}
\frac{g_{bw}(x_i)}{T_{bw}}\cdot log_2 \biggl(\frac{g_{bw}(x_i)}{T_{bw}}\biggr)
\end{equation*}
By replacing $H_{bw}(X)$ in Equation (\ref{eq:anondegree}), the anony\-mity
degree is, then, as follows:
\begin{equation*}
d_{bw} = - \displaystyle
\sum_{i=1}^{n}\frac{g_{bw}(x_i)}{T_{bw}\cdot log_2(n)}\cdot
log_2 \biggl(\frac{g_{bw}(x_i)}{T_{bw}}\biggr)
\end{equation*}

\begin{figure}[hptb]
    \centering
    \includegraphics[width=9.5cm]{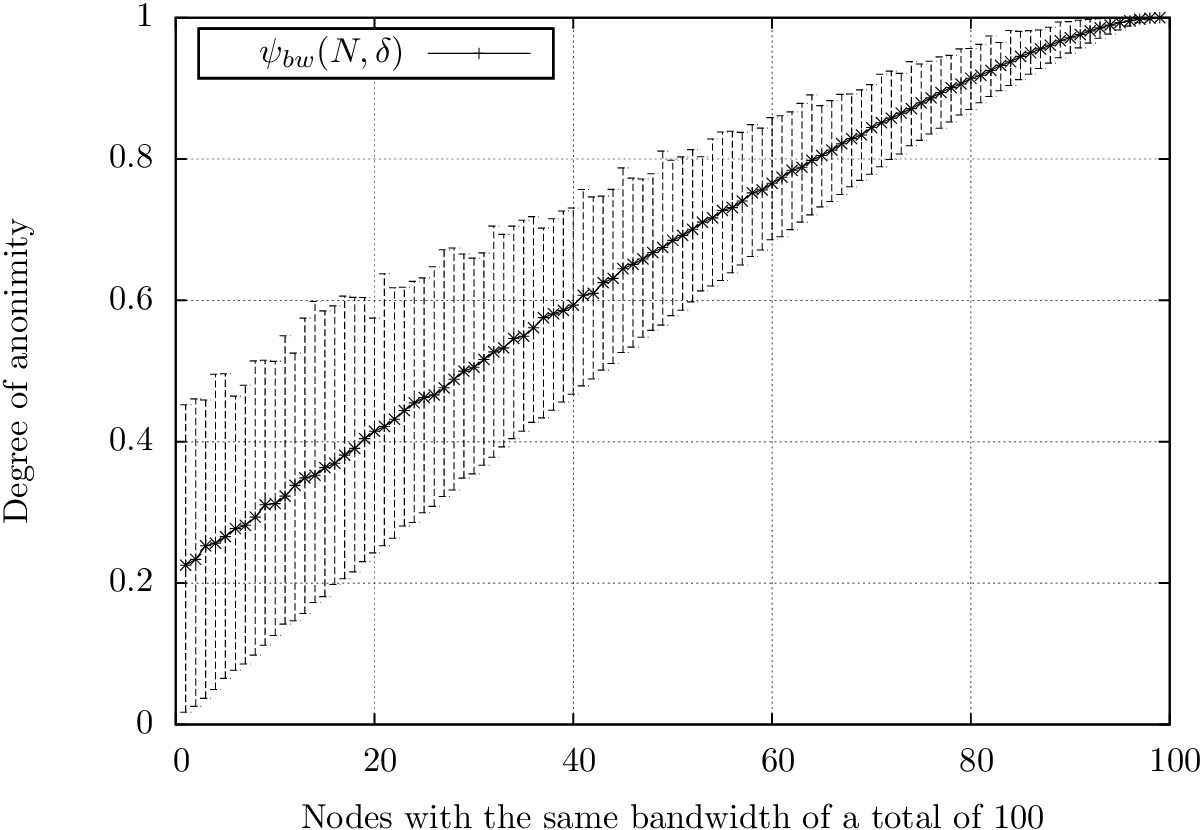}
    \caption{Influence of the uniformity of the bandwidth distribution in the
	     anonymity degree for $\psi_{bw}(N,\delta)$}
    \label{fig:bw}
\end{figure}

\begin{theorem}
Given a selection of Tor nodes $\psi_{bw}(N, \delta)\sim f_{bw}(x)$ with an
associated discrete random variable $X_{bw}$, the maximum anonymity degree is
achieved iff $g_{bw}(x_i)=K_{bw}$ $\forall x_i \in X_{bw}$, where $K_{bw}$ is a
constant.
\end{theorem}

\begin{proof}
($\Rightarrow$) $H(X_{bw}) = H_{M}(X_{bw})$ would imply that the
  anonymity degree gets maximum. This is only possible when
  $f_{bw}(x_i)= \frac{g_{bw}(x_i)}{T_{bw}}=\frac{1}{n}$, $\forall x_i \in
  X_{bw}$. Therefore,
\begin{align*}
   \frac{g_{bw}(x_i)}{T_{bw}} &= \frac{1}{n}\\
                  g_{bw}(x_i) &= \frac{T_{bw}}{n}
\end{align*}
and since $T_{bw}$ and $n$ are constant values for a certain Tor network, we can
consider that $g_{bw}(x_i)$ is also a constant, $\forall x_i \in X_{bw}$. \medskip

\noindent ($\Leftarrow$) Given $f_{bw}(x_i)=\frac{g_{bw}(x_i)}{T_{bw}}$ it is
easy to see that if $g_{bw}(x_i)=K_{bw}$ $\forall x_i \in X_{bw}$ then
$T_{bw}={\sum^{n}_{i=1} g_{bw}(x_i)}=n \cdot K_{bw}$ and, as a consequence,
$f_{bw}(x_i)=\frac{K_{bw}}{n \cdot K_{bw}} = \frac{1}{n}$ $\forall x_i \in
X_{bw}$. Hence, by replacing $f_{bw}(x_i)=\frac{1}{n}$ in Equation
(\ref{eq:anondegree}), we get $d_{bw}=1$.
\end{proof}

Figure \ref{fig:bw} shows the relation between the uniformity of the
bandwidth of the nodes and the anonymity degree of the whole system.
It depicts the anonymity degree of a Tor system with 100 nodes,
measured under different restrictions. In particular, the bandwidth of
the nodes has been modified in a manner that a certain subset of nodes
has the same bandwidth, and the bandwidth of the remainder nodes has
been fixed at random. During all the measurements the total
bandwidth of the system $T_{bw}$ remains constant. As the size of the
subset is increased, and more nodes have the same bandwidth, the
uncertainty is higher from the point of view of the discrete random
variable associated to $\psi_{bw}(N, \delta)$. Therefore, the
anonymity degree is increased when the uniformity of the distribution
of the bandwidths grows.

\section{New strategy based on latency graphs}
\label{sec:new-strategy}

We present in this section a new selection criteria. The new strategy
relies on modelling the Tor network as an undirected graph $G(V, E)$,
where $V=N \cup \{s\}$ denotes the set composed by the Tor nodes
$N=\{v_1,...,v_n\}$ and the client node $v_{n+1}=s$, and where
$E=\{e_{12}, e_{13}, ..., e_{ij}\}$ denotes the set of the edges of
the graph. We use the notation $e_{ij}=(v_i,v_j)$ to refer to the edge
between two nodes $v_i$ and $v_j$. The set of edges $E$ represents the
potential connectivity between the nodes in $V$, according to some
partial knowledge of the network status which the strategy has. If an
edge $e_{ij}=(v_i, v_j)$ is in $E$, then the connectivity between
nodes $v_i$ and $v_j$ is potentially possible. The set of edges $E$ is
a dynamic set, i.e., the network connectivity (from a TCP/IP
standpoint) changes periodically in time, while the set of vertices
$V$ is a static set. Finally, and although the network connectivity
from node $v_i$ to $v_j$ is not necessarily the same as the
connectivity from $v_j$ to $v_i$, we decided to model the graph as
undirected for simplicity reasons. Our decision also obeys to the two
following facts: (i) in a TCP/IP network, the presence of nodes is
more persistent than the connectivity among them; and (ii) the
connectivity is usually the same from a bidirectional routing point of
view in TCP/IP networks.

Related to the edges of the graph $G(V,E)$, we define a function $c_t:
E \rightarrow \mathbb{R} \cup \{\infty \}$ such that, for every edge
$e_{ij} \in E$, the function returns the associated network latency
between nodes $v_i$ and $v_j$ at time $t$. If there is no connectivity
between nodes $v_i$ and $v_j$ at time $t$, then we say that the
connectivity is undefined, and function $c_t$ returns the infinity
value. Notice that function $c_t$ can be implemented in several ways.
Some previous work in the field include software tools to monitor the
network based on IP geolocation \cite{dong2012}, modelling of networks
as stochastic systems \cite{crisos2012}, and network tomography
\cite{Coates2002}. Regardless of the strategy used to implement $c_t$,
there is an important restriction from a security point of view:
leakage of sensitive information in the measurement process shall be
contained. This mandatory constraint must always be fulfilled.
Otherwise, an adversary can benefit from a monitoring process in
order to degrade the anonymity degree.

\begin{algorithm}[!ht]
\caption{Latency Computation Process - lat\_comp$(G(V, E), \Delta t, m)$}
\label{alg:latencies}
\begin{algorithmic}
\STATE \textbf{Input:} $G(V, E), \Delta t, m$\\
\medskip
	\STATE $t_0 \leftarrow t_q \leftarrow 0$
	\STATE $E \leftarrow \varnothing$
	\STATE $L(e_{ij}) \leftarrow (\infty, t_0)$
	\medskip

	\WHILE{\textit{TRUE}}
	    \STATE $t_q \leftarrow t_q+1$

	    \FOR{$i \leftarrow 1$ \TO $m$}
		\STATE $i,j \leftarrow $random$(1,|V|)$, $i \neq j$
		\STATE $l_q \leftarrow c_t(e_{ij})$\\
		\medskip

		\IF {$l_q = \infty$}
		    \STATE $E \leftarrow E \setminus \{e_{ij}\}$
		\ELSE
		    \STATE $E \leftarrow E \cup \{e_{ij}\}$
		    \STATE Given $L(e_{ij})=(l_p, t_p)$
		    \IF {$l_p \neq \infty$}
			\STATE $\alpha \leftarrow (t_p-t_0)/(t_q-t_0)$
			\STATE $l_q  \leftarrow \alpha \cdot l_p + (1-\alpha) \cdot
				l_q$

		    \ENDIF
		    \STATE $L(e_{ij}) \leftarrow (l_q, t_q)$
		\ENDIF
	    \ENDFOR\\
	    \STATE sleep($\Delta t$)
	    \medskip
	\ENDWHILE
\end{algorithmic}
\end{algorithm}

\begin{algorithm}[!ht]
\caption{K-paths Computation Process - kpaths$(G(V, E), \delta, k, x\_node,
cur\_path, paths\_list)$}
\label{alg:kpaths}
\begin{algorithmic}
\STATE \textbf{Input:} $G(V,E), \delta, k, x\_node, cur\_path, paths\_list$\\
\medskip
	\IF {len($paths\_list$) = $k$}
		\STATE return
	\ENDIF
	\IF {len($cur\_path$) $> \delta$}
		\STATE return
	\ENDIF
	\medskip

	\STATE $v_l \leftarrow $last\_vertex($cur\_path$)
        \STATE $new\_len \leftarrow $len($cur\_path$)+1
	\STATE $adjacency\_list \leftarrow $adjacent\_vertices($G(V, E)$, $v_l$)
	\STATE remove\_nodes($adjacency\_list$, $cur\_path$)
	\STATE random\_shuffle($adjacency\_list$)

	\medskip

	\FOR{$vertex$ \textbf{in} $adjacency\_list$}
	      \IF {$vertex = x\_node$ \AND $new\_len$ $< \delta$}
		    \STATE \textbf{continue}
	      \ENDIF

	      \IF {$vertex = x\_node$ \AND $new\_len=\delta$}
		    \STATE $new\_sol \leftarrow cur\_path+ \langle vertex \rangle$
		    \STATE $paths\_list \leftarrow paths\_list +\langle new\_sol \rangle$
		    \STATE \textbf{break}
	      \ENDIF
	      \STATE $cur\_path \leftarrow cur\_path+ \langle vertex \rangle$
	      \STATE kpaths$(G(V, E), \delta, k, x\_node, cur\_path,
              paths\_list)$
	\ENDFOR
\end{algorithmic}
\end{algorithm}

\begin{algorithm}[!ht]
\caption{Graph of Latencies Selection of Nodes - $\psi_{grp}(N,\delta)$}
\label{alg:graph}
\begin{algorithmic}
\STATE \textbf{Input:} $G(V,E), s, \delta, k, max\_iter, \Delta t$\\
\STATE \textbf{Output:} $C=\langle s,e,r_{1},r_{2}, ... ,r_{\delta-2},x \rangle$,
$P=\{a1,...,a_\delta\}$\\
\medskip

	\STATE $P \leftarrow \varnothing$
	\STATE $paths\_list \leftarrow \langle \rangle$
	\STATE $iter \leftarrow 0$
	\medskip

	\emph{/* Executed in background as a process */}
	\STATE lat\_comp$(G(V, E), \Delta t, m)$
	\medskip

	\REPEAT
		\STATE $cur\_path \leftarrow \langle s \rangle$
		\STATE $x\_node \leftarrow $random\_vertex($V \setminus \{s\}$)
		\STATE kpaths$(G(V, E), \delta, k, x\_node, cur\_path,
                       paths\_list)$
		\STATE $iter \leftarrow iter+1$
	\UNTIL {(\NOT empty($paths\_list$)) \OR ($iter=max\_iter$)}
	\medskip

	\IF {\NOT empty($paths\_list$)}
		\STATE $C \leftarrow $min\_weighted\_path($paths\_list$)
	\ELSE
		\STATE $C \leftarrow $random\_path($V$, $\delta$)
	\ENDIF
	\FOR{$i \leftarrow 1$ \TO $\delta - 1$}
		\STATE $P \leftarrow P \cup \{(c_i, c_{i+1})\}$
	\ENDFOR
\end{algorithmic}
\end{algorithm}


Given the aforementioned rationale, we propose now the construction of
our new selection strategy by means of two general processes. A first
process computes and maintains the set of edges of the graph and its
latencies. The second process establishes, according to the outcomes
provided by the first process, circuit nodes. Circuit nodes are chosen
from those identified within graph paths with minimum latency. These
two processes are summarised, respectively, in Algorithms
\ref{alg:latencies} and \ref{alg:graph}. A more detailed explanation
of the proposed strategy is given below.

The first process (cf. Algorithm \ref{alg:latencies}) is executed in
background and keeps a set of labels related to each edge. Every label
is defined by the expression $L(e_{ij})=(l, t)$, where $e_{ij}$
denotes its associated edge. The label contains a tuple $(l, t)$
composed by an estimated latency $l$ between the nodes of the edge
(i.e., $v_i$ and $v_j$), and a time instant $t$ which specifies when
the latency $l$ was computed. When the process is executed for the
first time, the set of edges and all the labels are initialised as $E
\leftarrow \varnothing$ and $L(e_{ij}) \leftarrow (\infty, 0)$.

At every fixed interval of time $\Delta t$, the process associated to
Algorithm \ref{alg:latencies} proceeds indefinitely as follows. A set
of $m$ edges associated to the complete graph $K_n$ with the same
vertices of $G(V,E)$ are chosen at random. The latency associated to
every edge is estimated by means of the aforementioned function $c_t$.
If the computed latency is undefined (i.e., function $c_t$ returns the
infinity value), then the edge is removed from the set $E$ (if it was
already in $E$) and the associated latency labels not updated.
Otherwise, the edge is added to the set $E$ (if it was not already in
$E$), and the value of its corresponding labels updated. In
particular, the latency member of the tuple is modified by using a
exponentially weighted moving average (EWMA) strategy \cite{hunter86},
and the time member is updated according to the current time instant
$t_q$. For instance, let us suppose that we are in the time instant
$t_q$ and we have chosen randomly the edge $e_{ij}$ with an associated
label $L(e_{ij})=(l_p, t_p)$. Let us also suppose that
$l_q=c_{t_q}(e_{ij})$ is the new latency estimated for such an edge.
Thus, its corresponding label is updated according to the following
expression:
\begin{equation*}
L(e_{ij}) \leftarrow
\begin{cases}
  \bigl (l_p, t_p \bigr ),  & \text{if $l_q = \infty$;}\\
  \bigl (l_q, t_q \bigr )   & \text{if $l_p = \infty$;}\\
  \bigl (\alpha \cdot l_p + (1 - \alpha) \cdot l_q, t_q\bigr ) &
\text{otherwise}
\end{cases}
\end{equation*}
The first case of the previous expression corresponds to a situation
of disconnection between the nodes of the edge $e_{ij}$, and that has
been detected by the function $c_{t_q}$. As a consequence,
$c_{t_q}(e_{ij})$ returns infinity. In this case, the previous
estimated latency $l_p$ is maintained in the tuple, and the edge
$e_{ij}$ is removed from $E$. The second case can be associated to the
first time the latency of the edge $e_{ij}$ is estimated using
$c_{t_q}$, since the previous latency was undefined and the infinity
value is the one used in the first instantiation of $L(e_{ij})$. Under
the two last cases of the previous expression, the edge $e_{ij}$ is
always added to the set $E$ if it still does not belong to the
aforementioned set. The third scenario corresponds to the EWMA in the
strict sense. In this case, the coefficient $\alpha \in (0, 1)$
represents a smoothing factor. The value $\alpha$ has an important
effects in the resulting estimated latency stored in $L(e_{ij})$.
Notice that those values of $\alpha$ that are close to zero give a
greater weight to the recent measurements of the latency through the
function $c_{t_q}$. Contrary to this, a value of $\alpha$ closer to
one gives a greater weight to the historical measurements, making the
resulting latency less responsive to recent changes.

For the definition of the $\alpha$ factor we must consider that the
previous update of the latency ---for a certain edge--- could have
been performed long time ago. This is possible since, for every
interval of time $\Delta t$ we choose randomly just only $m$ edges to
update their latencies. Indeed, the value of $l_p$ in the previous
example could have been computed at the time instant $t_p$, and where
$t_p \ll t_q$. Therefore, if we define $\alpha$ as a static value, the
weight for previous measurements will always be the same,
independently of when the measurement was taken. This is not an
acceptable approach since the older the previous measurement is, the
less weight should have in the resulting computed latency.

To overcome this semantic problem, the coefficient $\alpha$ must be
defined as a dynamic value that takes into account the precise moment
in which the previous latencies were estimated for every edge. In
other words, $\alpha$ should be inversely proportional to the size of
the time interval between the previous measurement and the current
one. In order to define $\alpha$ as a function of this time interval,
we must keep the time instant of the previous latency estimation for a
given edge. This can be accomplished by storing the time instants in
the tuple of every edge label. Hence, every time we select at random
$m$ edges to update their latencies, its associated time members of
its labels must be updated with the current time instant $t_q$. It is
important to remark that this update process must be done just only
when the function $c_t$ returns a value different from the infinity
one. Moreover, for a selected edge $e_{ij}$ in the time instant $t_q$,
its $\alpha$ value is defined as:
\begin{equation*}
\alpha = \frac{t_p - t_0}{t_q - t_0}
\end{equation*}
where $t_0$ is the first time instant when the execution of the
process started. A graphical interpretation of the previous expression
is depicted in Figure \ref{fig:alphacoef}. We can appreciate that
$\alpha \in (0,1)$ by associating the numerator and the denominator of
the expression with its interval representation in the figure. Thus,
we can directly deduce that $0 < (t_p-t_0) < (t_q-t_0)$ and,
consequently, $\alpha \in (0,1)$. In this figure, we can also see the
influence of the previous time instant $t_p$ on the resulting
$\alpha$. In particular, three cases are presented: a) $t_p \ll t_q$,
b) $t_p \approx \frac{t_q-t_0}{2}$, and c) $t_p \approx t_q$. For
these cases, we can observe how $\alpha$ tends to, respectively, $0$,
$0.5$ and $1$.

\begin{figure}[hptb]
  \centering
  \subfloat[$\alpha \rightarrow 0$]{
	\includegraphics[width=0.25\textwidth]{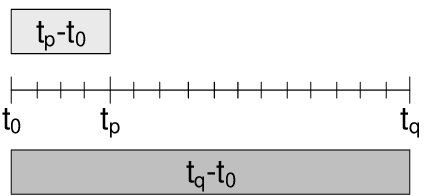}
	\label{fig:alpha0}}
  ~~\\
  \subfloat[$\alpha \rightarrow 0.5$]{
	\includegraphics[width=0.25\textwidth]{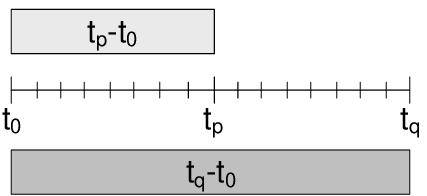}
	\label{fig:alpha05}}
  ~~\\
  \subfloat[$\alpha \rightarrow 1$]{
	\includegraphics[width=0.25\textwidth]{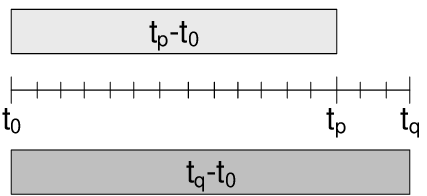}
        \label{fig:alpha1}}
  \caption{Graphical interpretation of the $\alpha$ coefficient}
  \label{fig:alphacoef}
\end{figure}

The second process (cf. Algorithm \ref{alg:graph}) is used for
selection of circuit nodes. It utilises the information maintained by
the process associated to Algorithm \ref{alg:latencies}. In
particular, the graph $G(V,E)$ and the labels $L(e_{ij})$ $\forall
e_{ij} \in E$ are shared between both processes. When a user wants to
construct a new circuit, this process is executed and it returns the
nodes of the circuit. For this purpose, an exit node $x$ is chosen at
random from the set of vertices $V \setminus \{s\}$. After that, the
process computes until $k$ random paths of length $\delta$ between the
nodes $s$ and $x$. With this aim, a recursive process, summarised
in Algorithm \ref{alg:kpaths}, is called. In the case that there is
not any path between the vertices $s$ and $x$, another exit node is
chosen and the procedure is executed again. This iteration must be
repeated until a) some paths of length $\delta$ between the pair of
nodes $s$ and $x$ are found, or b) until a certain number of
iterations are performed. In the first case, the path with the minimum
latency is selected as the solution among all the obtained paths. In
the second case, a completely random path of length $\delta$ is
returned. To avoid this situation, i.e., to avoid that our new
strategy behaves as a random selection of nodes strategy, the process
associated to Algorithm \ref{alg:latencies} must be started some time
before the effective establishment of circuits take place. This way,
the graph $G(V,E)$ increases the necessary level of connectivity among
its vertices. We refer to Section~\ref{sec:experiments} for more
practical details and discussions on this point.


\subsection{Discussion on the adversary model}

One may think that an adversary, as it was initially defined in
Section~\ref{sec:rationale}, can try to reconstruct the client graph
and guess the corresponding latency labels of our new strategy in
order to degrade its anonymity degree. However, even if we assume the
most extreme case, in which the adversary obtains a complementary
complete graph $K_n$ with the set of vertices $N$ and corresponding
latency labels, this does not affect the anonymity degree of our new
strategy. First of all, we recall that the graph of the client is a
dynamic random subgraph of $K_{n+1}$ that is evolving over time, with
a set of vertices $N \cup \{s\}$. The adversary graph would also be a
subgraph of $K_n$ with the set of vertices $N$, changing dynamically
as time goes by. Therefore, the set of vertices and edges of the
adversary and client graphs will never converge into same connectivity
model of the network. Moreover, the latencies between the client node
$s$ and any other potential entry node $e$ cannot be calculated by the
adversary. Otherwise, this would mean that the anonymity has already
been violated by the adversary. Indeed, the estimated latencies will
definitively differ between the client and the adversary graph, since
they are computed at different time frames and different source
networks. Finally, the adversary also ignores the exit nodes selected
by the client, as well as the $k$ parameter used by the client to
choose the paths.

\section{Analytical evaluation of the new strategy}
\label{sec:new-strategy-evaluation}

We provide in this section the analytical expression of the anonymity degree of
the new strategy. First, we extend the list of definitions provided in Section
\ref{sec:rationale}.

\subsection{Analytical graph of $\psi_{grp}(N, \delta)$}

In order to provide an analytical expression of the a\-no\-ny\-mi\-ty
degree it is important to notice that this must be always done from
the adversary standpoint. In this regard, the graph to be considered
for this purpose differs with respect to the one used to compute a
circuit. Note that the latencies associated to every edge which
contains the client node $s$ cannot be estimated by the adversary ---
specially if we consider that this particular node is unknown by the
adversary. Hence, an adversary aiming at violating the anonymity of
client node $s$ could try to estimate the user graph without node $s$
and its associated edges. This leads us to the following definition
(cf. Figure \ref{fig:latgraph} as a clarifying example):

\begin{definition}
Given a latency graph $G(V,E)$ associated to a selection of Tor nodes
$\psi_{grp}(N, \delta)$ strategy and the client node $s$, we define the
analytical graph as $G'(V',E')$ where $V'=V \setminus \{s\}$ and $E'=E \setminus
\{(s,v_i)\}$ $\forall v_i \in V$.
\end{definition}

\begin{figure}[tb]
  \centering
  \includegraphics[width=10cm]{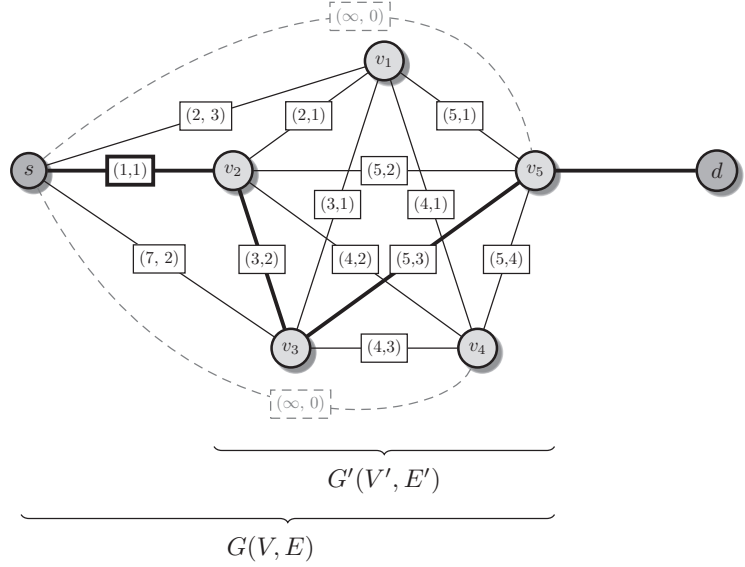}
  \caption{Example of a latency graph and its analytical graph with a
	   selected circuit $C=\langle s, v_2, v_3, v_5 \rangle$ of length
	   $\delta=3$}
  \label{fig:latgraph}
\end{figure}

\subsection{$\lambda$-betweenness and $\lambda$-betweenness probability}

For the purpose of computing the degree of anonymity of our new
strategy, a new metric inspired by the Freeman's \textit{betweenness
  centrality} measure \cite{Freeman77} is presented. This metric,
called $\lambda$-betweenness, is defined as a measurement of the
frequency which a node $v$ is traversed by all the possible paths of
length $\lambda$ in a graph. The formal definition is given below.

\begin{definition}
Consider an undirected graph $G(V, E)$. Let $KP_{st}$ denote the set of paths of
length $\lambda$ between a fixed source vertex $s \in V$ and a fixed target
vertex $t \in V$. Let $KP_{st}(v)$ be the subset of $KP_{st}$ consisting of
paths that pass through the vertex $v$. Then, we define the
$\lambda$-betweenness of the node $v \in V$ as follows:
\begin{equation*}
KP_{B}(v, \lambda)=\frac{\displaystyle \sum_{s,t \in V}\sigma_{st}(v, \lambda)}
{\displaystyle \sum_{s,t \in V} \sigma_{st}(\lambda)}
\end{equation*}
where $\sigma_{st}(\lambda) = |KP_{st}|$ and, $\sigma_{st}(v, \lambda)
= |KP_{st}(v)|$.
\end{definition}

As we can observe, the $\lambda$-betweenness provides the proportion between the
number of paths of length $\lambda$ which traverses a certain node $v$, and the
number of the total paths of length $\lambda$. However, since the degree of
anonymity needs a probability distribution, the following definition is
required.

\begin{definition}
\label{def:lambdaprob}
Consider an undirected graph $G(V, E)$. Let $KP_{B}(v, \lambda)$ be the
$\lambda$-betweenness of the node $v \in V$. Then, the $\lambda$-betweenness
probability of the node $v$ is defined as:
\begin{equation*}
LB(v, \lambda)=\frac{KP_{B}(v, \lambda)}{\displaystyle \sum_{w \in V}KP_{B}(w,
\lambda)}=\frac{\displaystyle \sum_{s,t \in V}\sigma_{st}(v,
\lambda)}{\displaystyle \sum_{w\in V} \sum_{s,t \in V}\sigma_{st}(w,
\lambda)}
\end{equation*}
\end{definition}
\noindent It follows immediately that $0 \leq LB(v, \lambda) \leq 1$,
$\forall v \in V$, since this expression is equivalent to the
normalised $\lambda$-betweenness.

\subsection{Entropy and anonymity degree}
\label{subsec:dgrp}
The graph of latencies selection of Tor nodes is defined formally as an
algorithm $\psi_{grp}(N, \delta) \sim f_{grp}(x)$ with an associated discrete
random variable $X_{grp}$ and an analytical graph $G'(V',E')$. The
$\mathrm{pmf}$ $f_{grp}(x)$ is given by means of the $\lambda$-between\-ness
probability expression:
\begin{equation*}
f_{grp}(x_i)=p_i=\mathbb{P}(X_{grp}=x_i)=\frac{\displaystyle \sum_{e,x \in
V'}\sigma_{ex}(v_i,\lambda)}{\displaystyle \sum_{w\in V'} \sum_{e,x \in
V'}\sigma_{ex}(w,\lambda)}
\end{equation*}
\noindent where $e$ and $x$ denotes every potential entry and exit
node respectively in a Tor circuit, and $\lambda=\delta - 1$. It is
worth noting that the value $\lambda=\delta - 1$ makes sense only if
we take into consideration that the client node $s$ and its edges are
removed in the analytical graph respect to the latency graph.

Hence, the entropy of a system whose clients use a graph of latencies selection
of nodes strategy is:
\begin{equation*}
H_{grp}(X) = - \displaystyle \sum_{i=1}^{n} LB(v_i, \lambda)\cdot log_2
\bigl(LB(v_i, \lambda)\bigr)
\end{equation*}
By replacing $H_{grp}(X)$ in Equation (\ref{eq:anondegree}), the degree of
anony\-mity is then:
\begin{equation*}
d_{grp} = - \displaystyle
\sum_{i=1}^{n}\frac{LB(v_i, \lambda)}{log_2(n)} \cdot log_2
\bigl(LB(v_i, \lambda)\bigr)
\end{equation*}

\begin{theorem}
\label{theo:graph1}
Given a selection of Tor nodes $\psi_{grp}(N, \delta)\sim f_{grp}(x)$ with an
associated discrete random variable $X_{grp}$ and an analytical graph
$G'(V',E')$ with $n=|V'|$ and $m=|E'|$, the a\-no\-ny\-mi\-ty degree is
increased as the density of the analytical graph grows.
\end{theorem}
\medskip

\begin{proof}
The density of a analytical graph $G'=(V',E')$ measures how many edges
are in the set $E'$ compared to the maximum possible number of edges
between vertices in the set $V'$. Formally speaking, the density is
given by the formula $\frac{2m}{n(n-1)}$. According to the previous
expression, and since the number of nodes of the analytical graph
remains constant, the only way to increase the density value is
through rising the value $m$; that is, by adding new edges to the
graph. Obviously, this implies that the more number of edges the
analytical graph has, the more its density value is augmented.

Moreover, if we increase the density of the analytical graph by adding
new edges, then the $\lambda$-betweenness probability of each vertex
will be affected. In particular, the denominator of the
$\lambda$-betweenness probability expression will change for all the
vertices in the same manner, whereas the numerator will be increased
for those vertices that lie on any new path of length $\lambda$ which
contains some of the added edges. However, this increase is not
arbitrary for a given vertex, since it has a maximum value determined
by the total amount of paths of length $\lambda$ which traverses such
vertex. Therefore, we can consider that each vertex has two states
while we are adding new edges. First, a transitory state where the
graph does not include all the paths of length $\lambda$ that traverse
such vertex. And second, a stationary state which implies that the
graph has all the paths of length $\lambda$ that traverses the given
vertex. Thus, if we add new edges at random, then the numerator of the
$\lambda$-betweenness probability of each vertex should be increased
uniformly. Consequently, the degree of anonymity grows when the
density of the graph is augmented.
\end{proof}

It is interesting to highlight that the numerator of the
$\lambda$-betweenness probability of a certain vertex will be
increased while it is in a transitory state, and until the vertex
achieves its stationary state. After that, such value cannot be
increased. It seems obvious that the degree of anonymity associated to
a particular analytical graph will be reached when all the vertices
are in a stationary states; or, in other words, when it is the
complete graph. Let us formalize this through the following theorem.

\begin{theorem}
\label{theo:graph2}
Given a selection of Tor nodes $\psi_{grp}(N, \delta)\sim f_{grp}(x)$ with an
associated discrete random variable $X_{grp}$ and an analytical graph
$G'(V',E')$ with $n=|V'|$, the maximum anonymity degree is achieved iff $G'(V',
E')$ is the complete graph $K_n$.
\end{theorem}
\medskip

\begin{proof}
($\Rightarrow$) Let us suppose that $G'(V', E')$ is not the complete graph
$K_n$. The maximum anonymity degree will be achieved when $LB(v_i, \lambda)$ is
equiprobable for all $v_i \in V'$. That is:
\begin{equation*}
\frac{\displaystyle \sum_{e,x \in V}\sigma_{ex}(v_i,
\lambda)}{\displaystyle \sum_{w\in V} \sum_{e,x \in V}\sigma_{ex}(w,
\lambda)} = \frac{1}{n}~~~~~\forall v_i \in V'
\end{equation*}
where $\lambda=\delta - 1$, and where $e$ and $x$ represents every possible
\textit{entry} and \textit{exit} node of a circuit respectively. The previous
expression can be rewritten as follows:
\begin{equation*}
{\displaystyle \sum_{e,x \in V'} \sigma_{ex}(v_i, \lambda)} =
\frac{\displaystyle \sum_{e,x \in V'}\sigma_{ex}(v_1,
\lambda) + ... + \sum_{e,x \in
V'}\sigma_{ex}(v_n,\lambda)}{n}
\end{equation*}
Let us now suppose that the value $\sum_{e,x \in V'} \sigma_{ex}(v_i,
\lambda)$ is fixed for every node of the analytical graph in
accordance to the previous expression. Then, since $G'(V',E')$ is not
the complete graph $K_n$, we can eliminate an arbitrary edge such that
the number of paths of length $\lambda$ with entry node $e$ and exit
node $x$, and which traverses a given particular node $v_j \in V'$, is
reduced. Thus, the value of $\sum_{e,x \in V'}\sigma_{ex}(v_j,
\lambda)$ would be affected for that given node. However, this
contradicts the previous expression, since $\sum_{e,x \in V'}
\sigma_{ex}(v_i, \lambda)$ would take different values for distinct
nodes, and when such value must be the same for any node of the graph.
\medskip

\noindent ($\Leftarrow$) Let us suppose, by contradiction, that the maximum
a\-no\-ny\-mity degree is not achieved by the analytical graph $K_n$ associated to
$\psi_{grp}(N, \delta)$. This implies that given two different nodes $v_j$ and
$v_k$ of the graph $K_n$, they will not have the same probability of being
chosen by $\psi_{grp}(N, \delta)$; that is, $LB(v_j, \lambda) \neq LB(v_k,
\lambda)$. Then, since $LB(v, \lambda)$ is defined as follows:
\begin{equation*}
LB(v, \lambda)=\frac{\displaystyle \sum_{e,x \in V'}\sigma_{ex}(v,
\lambda)}{\displaystyle \sum_{w\in V'} \sum_{e,x \in V'}\sigma_{ex}(w,
\lambda)}
\end{equation*}
We can consider that the only factor which makes possible the previous
restriction $LB(v_j, \lambda) \neq LB(v_k, \lambda)$ is in the
numerator, because the value of the denominator remains equal for both
nodes in a fixed graph. Thus, if we want to satisfy the previous
restriction, we must change the value $\sum_{e,x \in V'}\sigma_{ex}(v,
\lambda)$ of either node $v_j$ or node $v_k$. However, this is only
possible if we eliminate a particular edge of the graph. This
contradicts the imposed premise that the analytical graph associated
to $\psi_{grp}(N, \delta)$ was the complete graph $K_n$.
\end{proof}

Theorems \ref{theo:graph1} and \ref{theo:graph2} are exemplified in
conjunction in Figure \ref{fig:graph}. We can observe how a density
increase of an analytical graph influences in the degree of anonymity,
achieving its maximum value when the graph is the complete one (i.e.,
it has a density equal to one).

\begin{figure}[hptb]
    \centering
    \includegraphics[width=10cm]{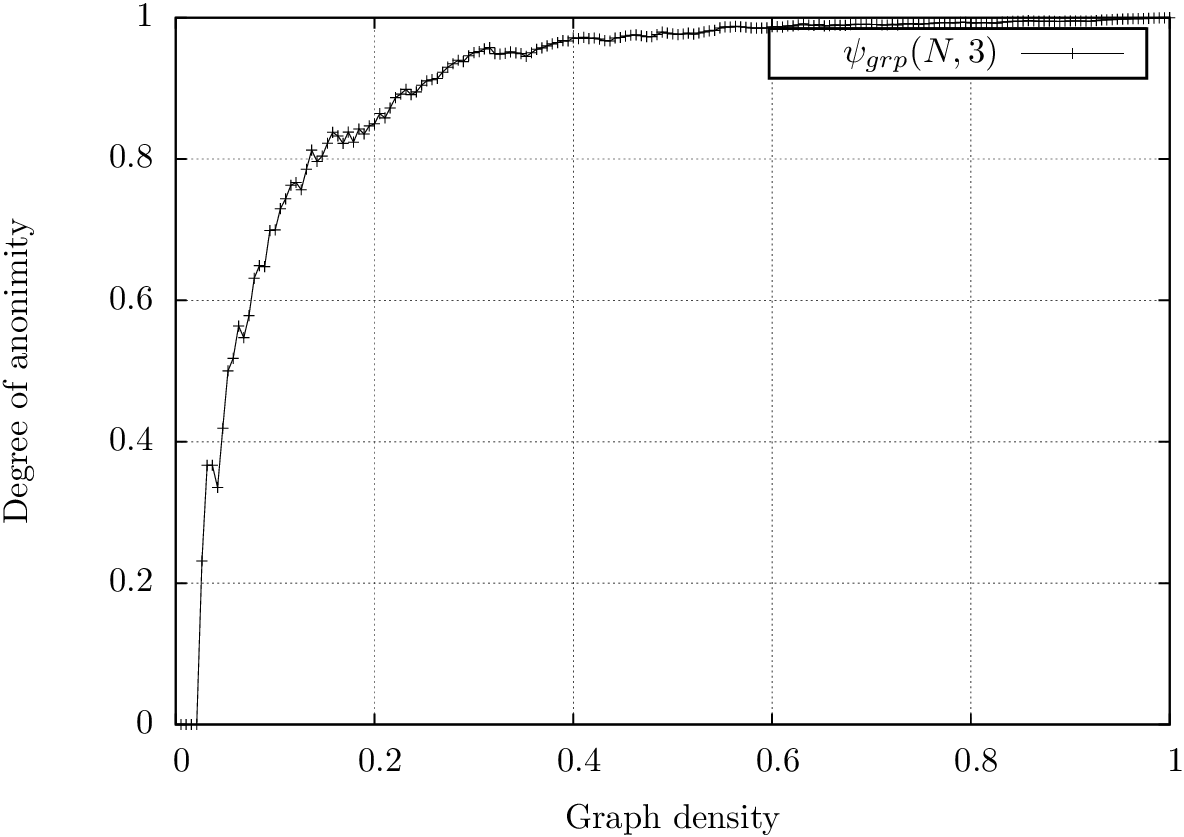}
    \caption{Influence of the density of the analytical graph in
the degree of anonymity with $|V'|=20$ and $\delta=3$}
    \label{fig:graph}
\end{figure}

\begin{theorem}
\label{theo:graph3}
Let $G(V,E)$ be a undirected graph with $n=|V|$ and let $\lambda$ be a fixed
length of a path, the value of $\sigma_{st}(\lambda)$ is maximised iff $G(V,E)$
is the complete graph $K_n$.
\end{theorem}
\medskip

\begin{proof}
($\Rightarrow$) Let us suppose, by contradiction, that $G(V,E)$ is not
  the complete graph $K_n$. Then, we can choose an arbitrary edge
  $e_{ij} \in E$ that belongs to a path of length $\lambda$ between
  the nodes $s$ and $t$. Then, we can remove $e_{ij}$ from $E$ since
  the graph is not complete. As a consequence, the value $KP_{st}$
  will be reduced. However, this contradicts the fact that the value
  $\sigma_{st}(\lambda)$ must be maximum since $\sigma_{st}(\lambda) =
  |KP_{st}|$.\\

\noindent ($\Leftarrow$) The proof is direct, since the complete graph $K_n$
contains all the possible edges between its nodes, and thus $KP_{st}$ consists
of all the possible paths of length $\lambda$ between the nodes $s$ and $t$.
\end{proof}
\medskip

\begin{theorem}
\label{theo:graph4}
Let $K_n$ be a complete graph, the total number of paths of length $\lambda$
between any pair of vertices $s$ and $t$ is given by the expression:

\begin{align*}
\sum_{s,t\in V} \sigma_{st}(\lambda) &= ((n-1)((n-1)^{\lambda} -
(-1)^{\lambda}))
\end{align*}
\end{theorem}

\begin{proof}
The proof is given in Appendix~\ref{sec:appendix2}.
\end{proof}

\begin{theorem}
\label{theo:graph5}
Given a selection of Tor nodes $\psi_{grp}(N, \delta)\sim f_{grp}(x)$ with an
associated discrete random variable $X_{grp}$ and an analytical graph
$G'(V',E')$, the maximum anony\-mity degree is achieved iff
\begin{align*}
\sum_{e,x\in V'} \sigma_{ex}(\lambda) & =
((n-1)((n-1)^{\lambda} - (-1)^{\lambda}))
\end{align*}
\end{theorem}

\begin{proof}
The proof is direct by applying Theorems \ref{theo:graph2},
\ref{theo:graph3} and \ref{theo:graph4}.
\end{proof}

\section{Experimental results}
\label{sec:experiments}

We present in this section a practical implementation and evaluation
of the series of strategies previously exposed. Each implementation
has undergone several tests, in order to evaluate latency penalties
during Web transmissions. Additionally, the degree of anonymity of
every experimental test is also estimated, for the purpose of drawing
a comparison among them.

\begin{table}[!b]
   \centering
   \begin{tabular}{|r|c|r||r|}
      \hline
      \multicolumn{3}{|c||}{\textbf{Real Tor network}} &
                            \textbf{PlanetLab}\\
      \hline
      $\#$ Nodes & Country & \% & $\#$ Nodes\\
      \hline
      \hline

      815 & US & 26.54 & 27 \\
      533 & DE & 17.36 & 17 \\
      187 & RU &  6.09 &  6 \\
      181 & FR &  5.89 &  6 \\
      171 & NL &  5.56 &  6 \\
      146 & GB &  4.75 &  5 \\
      132 & SE &  4.30 &  4 \\
       80 & CA &  2.61 &  3 \\
       56 & AT &  1.82 &  2 \\
       43 & AU &  1.40 &  1 \\
       40 & IT &  1.30 &  1 \\
       40 & UA &  1.30 &  1 \\
       39 & CZ &  1.27 &  1 \\
       38 & CH &  1.24 &  1 \\
       34 & FI &  1.11 &  1 \\
       34 & LU &  1.11 &  1 \\
       33 & PL &  1.08 &  1 \\
       32 & JP &  1.04 &  1 \\
       437 & Others ($<$1\%)& 14.23 & 15 \\
       \hline
       \hline
       3071 & -- &  100 & 100 \\
       \hline
   \end{tabular}
   \caption{Selected PlanetLab nodes per country according to the real
    Tor network distribution}
    \label{table:plabs_dist}
\end{table}

\subsection{Node distribution and configuration in PlanetLab}

In order to measure the performance of the strategies presented in our
work, some practical experiments have been conducted. In particular,
we deployed a private network of Tor nodes over the PlanetLab research
network \cite{Chun2003,Peterson2006}. Our deployed Tor network is
composed of 100 nodes following a representative distribution based on
the real (public) Tor network. We distributed the nodes of the private
Tor network following the public network distribution in terms of
countries and bandwidths. Table \ref{table:plabs_dist} summarises the
distribution values per country. The estimated bandwidths of the nodes
is retrieved through the directory servers of the real Tor
network \cite{tordirspecv3}. Then, we categorised the nodes according
to their bandwidths by means of the \textit{k-means clustering}
methodology \cite{macqueen67,anderberg73}. A value of $k=100$ is used
as the number of clusters (i.e., number of selected nodes in
PlanetLab). When the algorithm converges, a cluster is assigned
randomly to each node of the private Tor network. Subsequently, the
bandwidth of each node is configured with the value of its associated
centroid (i.e. the mean of the cluster). For such a purpose, the
directive \texttt{BandwidthRate} is used in the configuration file of
every node. Let us note that the country and bandwidth values are
considered as independent in the final node distribution
configuration. Indeed, there is no need to correlate both variables,
since the bandwidth of every node can be configured by its
corresponding administrator, while this fact does not depend on the
country which the node belongs to.

\begin{figure}[hptb]
    \centering
    \includegraphics[width=9.5cm]{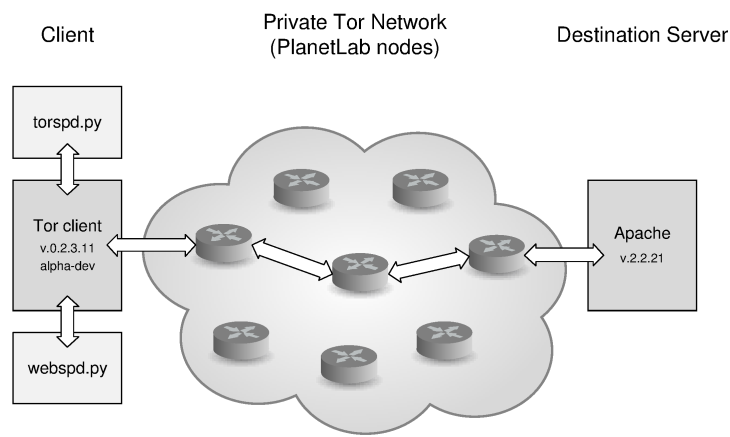}
    \caption{Conceptual representation of our testbed environment}
    \label{fig:tor-arch}
\end{figure}

\subsection{Testbed environment}
Every node of our Planetlab private network runs the Tor software,
version 0.2.3.11-alpha-dev. Additionally, four nodes inside the
network are configured as directory servers. These four nodes are in
charge of managing the global operation of the Tor network and
providing the information related to the network nodes.

\begin{figure}[hptb]
   \centering
   \subfloat[Web size of 50KB]{
      \includegraphics[width=8.25cm]{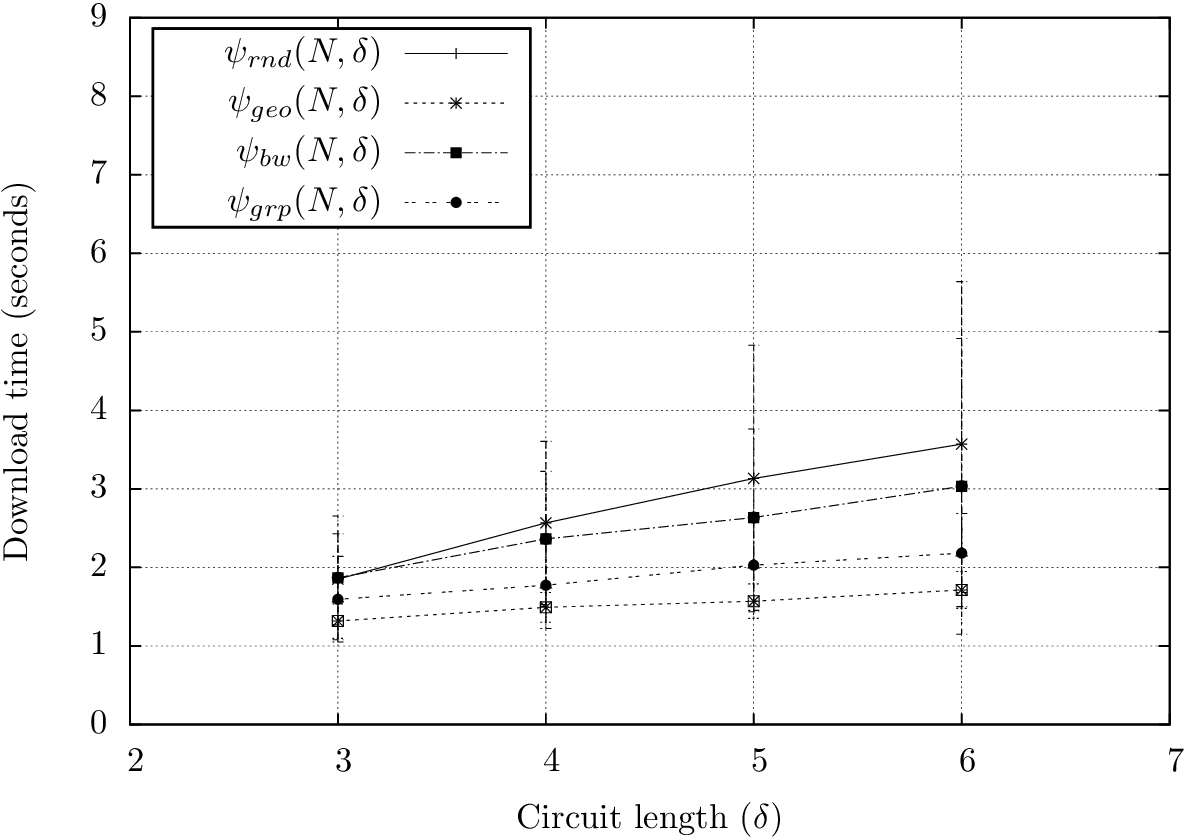}
      \label{fig:res-50}}
   ~~\\
   \subfloat[Web size of 150KB]{
      \includegraphics[width=8.25cm]{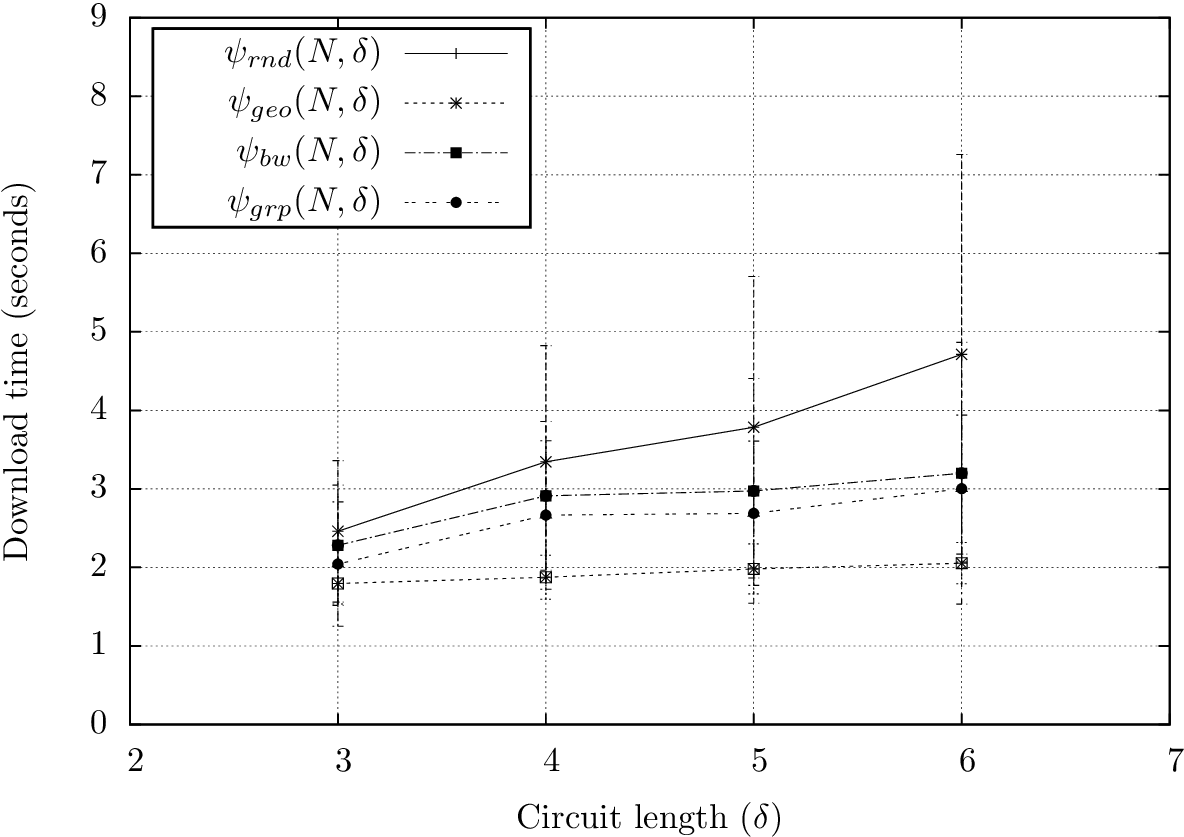}
      \label{fig:res-150}}
   ~~\\
   \subfloat[Web size of 320KB]{
      \includegraphics[width=8.25cm]{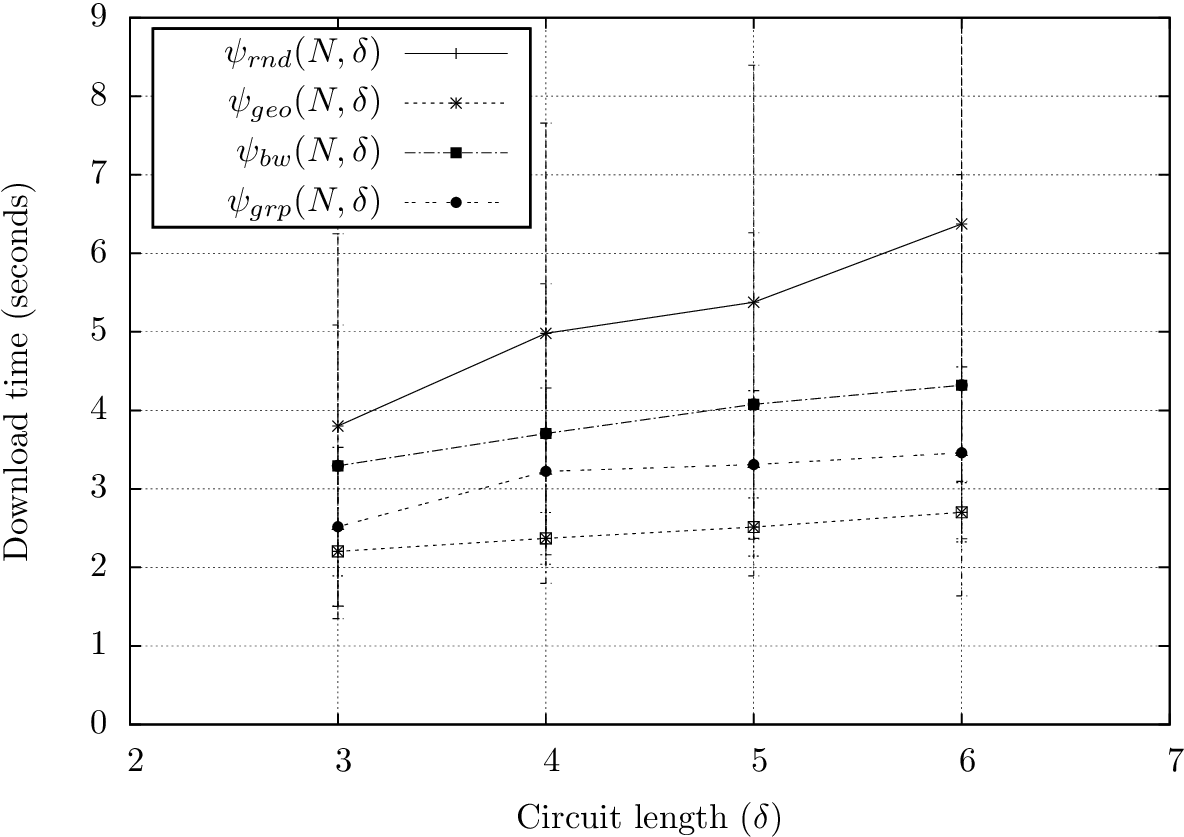}
      \label{fig:res-320}}
   \caption{Experimental results}
   \label{fig:expresults}
\end{figure}

Furthermore, two additional nodes outside the PlanetLab network are
used in our experiments. One of them is based on an Intel Core2 Quad
Processor at 2.66GHz with 6GB of RAM and a Gentoo GNU/Linux Operating
System with a 3.2.9 kernel. This one is used as the \textit{client}
node who handles the construction of Tor circuits for every evaluated
strategy. For this purpose, this node runs also our own specific
software application, hereinafter denoted as \texttt{torspd.py}. A
beta release of \texttt{torspd.py}, written in Python 2.6.6, can be
downloaded at \url{http://github.com/sercas/torspd}. The
\texttt{torspd.py} application relies on the TorCtl Python bindings
\cite{Mathewson2011} ---a Tor controller software to support path
building and various constraints on node and path selection, as well
as statistic gathering. Moreover, \texttt{torspd.py} also benefits
from the package NetworkX \cite{Hagberg2008} for the creation,
manipulation, and analysis of graphs. The client node is not only in
charge of the circuit construction given a certain strategy, but also
of attaching an initiated HTTP connection to an existing circuit. To
accomplish this, the node uses \texttt{torspd.py} to connect to an
special port of the local Tor software called the \textit{control
  port}, and which allows to command the operations. The client node
includes an additional software ---also based on Python--- capable of
performing HTTP queries through our private Tor network by using a
SOCKS5 connection against the local Tor client. This software, called
\texttt{webspd.py}, is also able to obtain statistics results about
the launched queries in order to evaluate the performance of the
algorithms implemented in \texttt{torspd.py}. Finally,
\texttt{webspd.py} performs every HTTP query making use directly of
the IP address of the destination server; consequently, any
perturbation introduced by a DNS resolution is avoided in our
measurements. The second node outside the PlanetLab network is based
on an Intel Xeon Processor at 2.00GHz with 2GB of RAM and a Debian
GNU/Linux Operating System with a 2.6.26 kernel. This node is
considered as the \textit{destination server}, and includes an HTTP
server based on Apache, version 2.2.21. The conceptual infrastructure
used to carry out our experiments is illustrated in Figure
\ref{fig:tor-arch}.

With the purpose of obtaining extrapolative results, we consider in
our testbed the outcomes reported in \cite{Ramachandran2010}. This
report, based on the analysis of more than four billion Web pages,
provides estimations of the average size of current Internet sites, as
well as the average number of resources per page and other interesting
metrics. Our testbed is built bearing in mind these premises, so that
it is close enough to a real Web environment. This way, the analysed
strategies (i.e., random selection, geographical selection, bandwidth
selection, and graph of latencies selection) are evaluated based on
three different series of experiments that vary the Web page sizes.
More precisely, the client node requests via our private PlanetLab Tor
network Web pages of, respectively, 50KB, 150KB and 320KB of size
---being the last one the average size of a Web page according to the
aforementioned report. The length of the circuits is seen as another
variable in our testbed. More precisely, the different strategies are
evaluated with Tor circuits of length three, four, five and six. Every
experiment is repeated 100 times, from which we obtain the minimum,
maximum and average time needed to download the corresponding Web
pages. Likewise, the standard deviation is computed for every test.
The obtained numerical results are presented in
Tables~\ref{table:exp_rnd},~\ref{table:exp_geo},~\ref{table:exp_bw}
and~\ref{table:exp_graph}, and also depicted graphically in
Figure~\ref{fig:expresults}. In the sequel, we use these results to
analyse the performance of every strategy in terms of transmission
times and degree of anonymity.

\subsection{Random selection of nodes strategy evaluation}
\label{sec:rnd_evaluation}
As previously exposed in Theorem \ref{theo:rnd}, the random selection
of nodes strategy is the best one from the point of view of the degree
of anonymity, since it achieves the maximum possible value.
Nevertheless, this selection of nodes methodology suffers from an high
penalty in terms of latency in accordance with the extrapolated
results of our evaluation. As it can be inferred from the analysis of
the numerical outcomes, and reflected in Figure \ref{fig:expresults},
the random selection algorithm exhibits the worst transmission times,
regardless of the size of the site or the length of the circuit used.
This can be explained by the random nature of this strategy. Indeed,
by selecting the nodes at random, the strategy can incur in some
problems which affect directly to the latency of a computed circuit,
such as a big distance between the involved nodes (in terms of
countries, i.e., routers), a network congestion in a part of the
circuit \cite{huali2012}, or a selection of nodes with limited
computational resources, among others. It is clear that all these
drawbacks are hidden to the strategy and explain the obtained results.
Moreover, all these problems are reflected in the standard deviation
of the measurements, which is the higher one compared with the other
alternatives.

\begin{table}[hptb]
\centering
\begin{tabular}{|c|c|c|c|c|}
  \hline
   \multicolumn{5}{|c|}{$\psi_{rnd}(N, \delta)$, $d_{rnd} = 1.0$, Web size 50KB} \\\hline
   Circ. length & Min. & Max. & Avg. & Std. dev. \\\hline
   $\delta=3$ & 0.95094203949 & 3.38077807426 & 1.84956678152 & 0.58107725003 \\\hline
   $\delta=4$ & 1.14792490005 & 7.46992301941 & 2.56735023022 & 1.03927644851 \\\hline
   $\delta=5$ & 1.13161778450 & 12.7252390385 & 3.13187572718 & 1.69722167190 \\\hline
   $\delta=6$ & 1.57145905495 & 14.6901309490 & 3.56973065615 & 2.06960596616
   \\\hline

   \multicolumn{5}{|c|}{$\psi_{rnd}(N, \delta)$, $d_{rnd} = 1.0$, Web size 150KB} \\\hline
   Circ. length & Min. & Max. & Avg. & Std. dev. \\\hline
   $\delta=3$ & 0.970992088318 & 5.70451307297 & 2.46016269684 & 0.901269612931 \\\hline
   $\delta=4$ & 1.081045866010 & 12.0326070786 & 3.34545367479 & 1.478886535440  \\\hline
   $\delta=5$ & 1.624027013780 & 16.0551090240 & 3.78437126398 & 1.918732505410 \\\hline
   $\delta=6$ & 2.279263019560 & 11.5805990696 & 4.71352141102 & 2.544477101520
   \\\hline

   \multicolumn{5}{|c|}{$\psi_{rnd}(N, \delta)$, $d_{rnd} = 1.0$, Web size 320KB} \\\hline
   Circ. length & Min. & Max. & Avg. & Std. dev. \\\hline
   $\delta=3$ & 1.49153804779 & 13.2033219337 & 3.79921305656 & 2.45165379541 \\\hline
   $\delta=4$ & 1.84271001816 & 15.2616338730 & 4.98011079788 & 2.67792560196 \\\hline
   $\delta=5$ & 1.73619008064 & 17.1969499588 & 5.37626729012 & 3.01781647919 \\\hline
   $\delta=6$ & 2.16737580299 & 17.8402540684 & 6.37420113325 & 3.27889183837 \\\hline
\end{tabular}
\caption{Random selection of nodes strategy ($\psi_{rnd}$) results}
\label{table:exp_rnd}
\end{table}

\begin{table}[hptb]
\centering
\begin{tabular}{|c|c|c|c|c|}
  \hline
   \multicolumn{5}{|c|}{$\psi_{geo}(N, \delta)$, $d_{geo} \approx 0.7157$, Web size 50KB} \\\hline
   Circ. length & Min. & Max. & Avg. & Std. dev. \\\hline
   $\delta=3$ & 0.913872003555 & 2.36748099327 & 1.31694087505 & 0.219359721740 \\\hline
   $\delta=4$ & 1.083739995960 & 2.03739213943 & 1.49165359974 & 0.189194613865 \\\hline
   $\delta=5$ & 1.157481908800 & 2.17184281349 & 1.56993633509 & 0.220167861127 \\\hline
   $\delta=6$ & 1.200492858890 & 2.63958501816 & 1.71368015051 & 0.234977785757
   \\\hline

   \multicolumn{5}{|c|}{$\psi_{geo}(N, \delta)$, $d_{geo} \approx 0.7157$, Web size 150KB} \\\hline
   Circ. length & Min. & Max. & Avg. & Std. dev. \\\hline
   $\delta=3$ & 1.38168692589 & 2.68786311150 & 1.79467165947 & 0.260276001481 \\\hline
   $\delta=4$ & 1.27939105034 & 2.92536497116 & 1.87463890314 & 0.281488220772 \\\hline
   $\delta=5$ & 1.33843898773 & 3.71059083939 & 1.98130603790 & 0.318113252410 \\\hline
   $\delta=6$ & 1.40922594070 & 3.28039193153 & 2.05482839346 & 0.261217096578
   \\\hline

   \multicolumn{5}{|c|}{$\psi_{geo}(N, \delta)$, $d_{geo} \approx 0.7157$, Web size 320KB} \\\hline
   Circ. length & Min. & Max. & Avg. & Std. dev. \\\hline
   $\delta=3$ & 1.41799902916 & 2.93465995789 & 2.20432470083 & 0.310828573513 \\\hline
   $\delta=4$ & 1.54156398773 & 3.33606600761 & 2.37035997391 & 0.329438846284 \\\hline
   $\delta=5$ & 1.88031601906 & 4.10431504250 & 2.51430423737 & 0.370494801277 \\\hline
   $\delta=6$ & 1.64570999146 & 3.89323496819 & 2.70262962818 & 0.376313686885 \\\hline
\end{tabular}
\caption{Geographical selection strategy ($\psi_{geo}$) results}
\label{table:exp_geo}
\end{table}

\begin{table}
\centering
\begin{tabular}{|c|c|c|c|c|}
  \hline
   \multicolumn{5}{|c|}{$\psi_{bw}(N, \delta)$, $d_{bw} \approx 0.9009$, Web size 50KB} \\\hline
   Circ. length & Min. & Max. & Avg. & Std. dev. \\\hline
   $\delta=3$ & 0.964261054993 & 5.12318110466 & 1.86709306002 & 0.789060168081 \\\hline
   $\delta=4$ & 1.078310012820 & 5.41474699974 & 2.36407416582 & 0.859666129425 \\\hline
   $\delta=5$ & 1.060457944870 & 6.92380499840 & 2.63418945789 & 1.128347022810 \\\hline
   $\delta=6$ & 1.278292894360 & 12.7536408901 & 3.03272451162 & 1.882337407440

   \\\hline

   \multicolumn{5}{|c|}{$\psi_{bw}(N, \delta)$, $d_{bw} \approx 0.9009$, Web size 150KB} \\\hline
   Circ. length & Min. & Max. & Avg. & Std. dev. \\\hline
   $\delta=3$ & 1.26475811005 & 7.09091401100 & 2.28234255314 & 0.765374484505 \\\hline
   $\delta=4$ & 1.23797798157 & 6.80870413780 & 2.91089500189 & 0.947719280103 \\\hline
   $\delta=5$ & 1.45632719994 & 12.6443610191 & 2.97445464373 & 1.431690789930 \\\hline
   $\delta=6$ & 1.27809882164 & 12.7246098518 & 3.19875429869 & 1.666334473980

   \\\hline

   \multicolumn{5}{|c|}{$\psi_{bw}(N, \delta)$, $d_{bw} \approx 0.9009$, Web size 320KB} \\\hline
   Circ. length & Min. & Max. & Avg. & Std. dev. \\\hline
   $\delta=3$ & 1.49932813644 & 12.9250459671 & 3.29500451326 & 1.79104222251 \\\hline
   $\delta=4$ & 1.52931094170 & 13.7227480412 & 3.70603173733 & 1.90767488259 \\\hline
   $\delta=5$ & 1.66296601295 & 17.3828690052 & 4.07738301039 & 2.18405609668 \\\hline
   $\delta=6$ & 2.04065585136 & 20.1761889458 & 4.32070047140 & 2.68160673888
   \\\hline
\end{tabular}
\caption{Bandwidth selection strategy ($\psi_{bw}$) results}
\label{table:exp_bw}
\end{table}

\begin{table}
\centering
\begin{tabular}{|c|c|c|c|c|}
  \hline
   \multicolumn{5}{|c|}{$\psi_{grp}(N, \delta)$, $d_{grp}$ (c.f. Section \ref{sec:graph_evaluation}), Web size 50KB} \\\hline
   Circ. length & Min. & Max. & Avg. & Std. dev. \\\hline
   $\delta=3$ & 0.935021877289 & 3.61296200752 & 1.59488223791 & 0.545028374794 \\\hline
   $\delta=4$ & 0.998504877090 & 3.74897003174 & 1.77225045919 & 0.548956074123 \\\hline
   $\delta=5$ & 1.195134162900 & 4.21774697304 & 2.02931211710 & 0.576776679346 \\\hline
   $\delta=6$ & 1.267808914180 & 3.35924196243 & 2.18245174408 & 0.502899662482
   \\\hline

   \multicolumn{5}{|c|}{$\psi_{grp}(N, \delta)$,  $d_{grp}$ (c.f. Section \ref{sec:graph_evaluation}). Web size 150KB} \\\hline
   Circ. length & Min. & Max. & Avg. & Std. dev. \\\hline
   $\delta=3$ & 1.112107038500 & 5.53429508209 & 2.04227621531 & 0.790901626275 \\\hline
   $\delta=4$ & 1.290552854540 & 5.68215894699 & 2.66674958944 & 0.944641197284 \\\hline
   $\delta=5$ & 1.163586854930 & 7.41387891769 & 2.68937173843 & 0.917799034111 \\\hline
   $\delta=6$ & 1.550453186040 & 5.40707683563 & 3.00299987316 & 0.935654846647
   \\\hline

   \multicolumn{5}{|c|}{$\psi_{grp}(N, \delta)$,  $d_{grp}$ (c.f. Section \ref{sec:graph_evaluation}), Web size 320KB} \\\hline
   Circ. length & Min. & Max. & Avg. & Std. dev. \\\hline
   $\delta=3$ & 1.502956867220 & 7.29033994675 & 2.51847231626 & 1.009688576850 \\\hline
   $\delta=4$ & 1.498482227330 & 6.52234792709 & 3.22330027342 & 1.061893260420 \\\hline
   $\delta=5$ & 1.734797000890 & 6.73247194290 & 3.31047295094 & 0.940285391625 \\\hline
   $\delta=6$ & 1.689666986470 & 7.89933013916 & 3.46063615084 & 1.094579395080

   \\\hline
\end{tabular}
\caption{Graph of latencies selection strategy ($\psi_{grp}$) results}
\label{table:exp_graph}
\end{table}

\subsection{Geographical selection of nodes strategy evaluation}
\label{sec:geo_evaluation}

The evaluation of the geographical selection of nodes strategy has
been performed by fixing the country and taking into consideration the
node distribution detailed in Table \ref{table:plabs_dist}. United
States was selected in accordance to the country where the client node
resides. Therefore, we can calculate the anonymity degree for this
strategy by recalling its related expression introduced in Section
\ref{sec:geoip_selection}:
\begin{equation*}
d_{geo} = \frac{log_2(m)}{log_2(n)}=\frac{log_2(27)}{log_2(100)}\approx 0.7157
\end{equation*}
As we can observe, the degree of anonymity has dropped significantly
when we compare it with the results of the other strategies. However,
sacrificing a certain level of anonymity incurs in a drastic fall of
the latency needed to download a Web page, as it can be noticed if we
compare Figures \ref{fig:res-50}, \ref{fig:res-150} and
\ref{fig:res-320}. In fact, this selection of nodes methodology
provides the best performance in terms of the time required to
download a Web page among the other alternatives. It is also
interesting to remark the fact that the standard deviation of the time
measured in this method remains nearly constant regardless of the
circuit length and the size of the Web page. This seems reasonable
since the more geographically near are the nodes, the less random
interferences affect to the whole latency. We can understand this if
we think in terms of the number of networks elements (i.e., routers,
switches, etc.) involved in the TCP/IP routing process between every
pair of nodes. Thus, a pair of nodes which belong to the same country
will be interconnected through less network elements compared to two
nodes which belong to different countries and, as a consequence, the
latency will be more stable along time. This can be an interesting
fact, since the penalty introduced by the use of Tor affects less to
the psychological perception of the user when browsing the Web
\cite{Kopsell2006}. Nevertheless, the anonymity degree of this
strategy is strongly tied to the fixed country, since ---as we pointed
out in Theorem \ref{theo:bw2}--- the less nodes belonging to the
country, the less anonymity degree is provided.

\subsection{Bandwidth selection of nodes strategy evaluation}
\label{sec:bw_evaluation}

The anonymity degree of the bandwidth selection of nodes strategy has
been computed empirically according to its associated formula (cf.
Section \ref{sec:bw_selection} for details). In particular, the
\texttt{torspd.py} application was in charge of obtaining the
bandwidth of every node of our private Tor network and of calculating
the anonymity degree. Thus, the anonymity degree when the evaluation
of this strategy was performed was approximately 0.9009. It is
important to highlight that, in spite of the fixed bandwidth specified
in the configuration, the bandwidth of every onion router is estimated
periodically by the Tor software running at every node, and provided
later to \texttt{torspd.py} through the directory servers. Indeed, if
we think that the established bandwidth of a node through its
configuration does not necessarily correspond to the real value, then
the anonymity degree can change in time in comparison to the previous
strategies.

From the viewpoint of the latency results, we can observe how the
bandwidth selection of nodes strategy improves the values respect to
the random strategy by sacrificing some degree of anonymity. However,
it does not achieve the transmission times of the geographical
methodology. The reason for that is because this strategy does not
take into account important networking aspects, such as network
congestion, number of routers, etc., that also impact the transmission
times. Therefore, it is fairly reasonable that this methodology is
more susceptible to networking problems, resulting in an increase of
the eventual transmission time results. This is also corroborated by
the standard deviation results, noting the lack of stability of
the results. In fact, the transmission times increase as the size
of the Web page or the length of the circuit also increase.

\subsection{Graph of latencies strategy evaluation}
\label{sec:graph_evaluation}
The experimental evaluation of our proposal has been performed after
the establishment of the parameters of its related algorithms. In
particular, they were $\Delta t=5$, $m=3$, $k=300$ and $max\_iter=5$.
Furthermore, the \textit{Latency Computation Process} was launched two
hours before the execution of \texttt{webspd.py}, leading to an
analytical graph with a set of more than 3,000 edges, and which
represents a density value of, approximately, 0.67. At this moment,
the \texttt{torspd.py} estimated the degree of anonymity in accordance
to the formula presented in Section \ref{subsec:dgrp}. Since such
equation depends on the length of the circuit, the anonymity degree
was estimated for lengths 3, 4, 5 and 6, giving the results of 0.9987,
0.9984, 0.9982 and 0.9981, respectively. As occurs with the previous
strategy, the degree of anonymity is dynamic over time, and in this
case depends on the connectivity of the analytical graph.
Nevertheless, the anonymity degree was not estimated again during the
evaluation tests.

Function $c_t$ was implemented by means of the construction of random
circuits of length $m$. Such circuits are not used as anonymous
channels for Web transmissions, but to estimate the latencies of the
edges. This is possible since during the construction of a circuit,
every time a new node is added to the circuit, the \textit{Latency
  Computation Process} is notified. Hence, it is easy to determine the
latency of an edge by subtracting the time instants of two nodes
added consecutively to a certain circuit. Regarding this \textit{modus
  operandi} of measuring the latencies, it is interesting to highlight
two aspects. The first one is that it meets the restriction of
estimating the latencies secretly; and the second one is that it not
only measures the latencies in relation the network solely, but also
takes into consideration delays motivated by the status of the nodes
or its resources limitations. This way, our proposal models indirectly
some negative issues which the other strategies do not reflect,
leading to an improvement of the transmission times as the obtained
results evidence.

By comparing the results of the previous strategies with the current
one, we can observe how our new proposal exhibits a better trade-off
between degree of anonymity and transmission latency. Particularly,
from the perspective of the transmission times, our proposal is quite
close to those from the geographical selection strategy, while it
provides a higher degree of anonymity. Indeed, if we compare our
strategy from the anonymity point of view, we can observe that only
the random selection of nodes criteria overcomes our new strategy,
but, as already mentioned, by sacrificing considerably the
transmission time performance.

\section{Related Work}
\label{sec:sota}

The use of entropy-based metrics to measure the anony\-mity degree of
infrastructures like Tor was simultaneously established by Diaz \textit{et al.}
\cite{Diaz2002} and Serjantov and Danezis \cite{Danezis2002}. Since then,
several other authors have proposed alternative measures \cite{Hamel2011}.
Examples include the use of the min entropy by Shmatikov and Wang in
\cite{Shmatikov2006}, and the Renyi entropy by Clau{\ss} and Schiffner in
\cite{Clauss2006}. Other examples include the use of combinatorial measures by
Edman \textit{et al.} \cite{Edman2007}, later improved by Troncoso \textit{et
al.} in \cite{Troncoso2008}. Snader and Borisov proposed in \cite{Snader2008}
the use of the Gini coefficient, as a way to measure inequalities in the circuit
selection process of Tor. Murdoch and Watson propose in \cite{Murdoch2008} to
asses the bandwidth available to the adversary, and its effects to degrade the
security of several path selection techniques.

With regard to literature on selection algorithms, as a way to improve
the anonymity degree while also increasing performance, several
strategies have been reported. Examples include the use of
reputation-based strategies \cite{Bauer2007}, opportunistic weighted
network heuristics \cite{Snader2008,Snader2009}, game theory
\cite{Zhang2010}, and system awareness \cite{Edman2009}. Compared to
those previous efforts, whose goal mainly aim at reducing overhead via
bandwidth measurements while addressing the classical threat model of
Tor \cite{Syverson2001}, our approach takes advantage of latency
measurements, in order to best balance anonymity and performance.
Indeed, given that bandwidth is simply self-reported on Tor, regular
nodes may be mislead and their security compromised if we allow nodes
from using fraudulent bandwidth reports during the construction of Tor
circuits \cite{Bauer2007,setop2008}.

The use of latency-based measurements for path selection on anonymous
infrastructures has been previously reported in the literature. In
\cite{Sherr2009}, Sherr \textit{et al.} propose a link-based path
selection strategy for onion routing, whose main criterion relies, in
addition to bandwidth measures, on network link characteristics such
as latency, jitter, and loss rates. This way, false perception of
nodes with high bandwidth capacities is avoided, given that
low-latency nodes are now discovered rather than self-advertised.
Similarly, Panchenko and Renner \cite{Panchenko2009} propose in their
work to complement bandwidth measurements with ro\-und trip time during
the construction of Tor circuits. Their work is complemented by
practical evaluations over the real Tor network and demonstrate the
improvement of performance that such latency-based strategies achieve.
Finally, Wang \textit{et al.} \cite{Wang2011,Wang2012} propose the
use of latency in order to detect and prevent congested nodes, so that
nodes using the Tor infrastructure avoid routing their traffic over
congested paths. In contrast to these proposals, our work aims at
providing a defence mechanism. Our latency-based approach is
considered from a node-centred perspective, rather than a
network-based property used to balance transmission delays. This way,
adversarial nodes are prevented from increasing their chances of
relying traffic by simply presenting themselves as low-latency nodes,
while guaranteeing an optimal propagation rate by the remainder nodes
of the system.

\section{Conclusion}
\label{sec:conclusion}
We addressed in this paper the influence of circuit construction strategies on
the anonymity degree of the Tor ({\em The onion router}) anonymity
infrastructure. We evaluated three classical strategies, with respect to their
de-anonymisation risk and latency, and regarding its performance for anonymising
Internet traffic. We then presented the construction of a new circuit selection
algorithm that considerably reduces the success probability of linking attacks
while providing enough performance for low-latency services. Our experimental
results, conducted on a real-world Tor deployment over PlanetLab confirm the
validity of the new strategy, and shows that it overperforms the classical ones.

\appendix

\section{Number of walks of length $\lambda$ between any two distinct vertices
of a $K_n$ graph}
\label{sec:appendix2}

Let $K_n$ be a complete graph with $n$ vertices and $\frac{n(n-1)}{2}$ edges,
such that every pair of distinct vertices is connected by a unique edge. Then, a
walk in $K_n$ of length $\lambda$ from vertex $v_1$ to vertex $v_{\lambda +1}$
corresponds to the following sequence:
\begin{center}
$\underbrace{v_1 \xrightarrow{e_1} v_2 \xrightarrow{e_2} v_3
    \xrightarrow{e_3} v_4 \xrightarrow{e_4}... \xrightarrow{e_{\lambda
        -1}} v_{\lambda} \xrightarrow{e_{\lambda}}
    v_{\lambda+1}}_{\text{walk in G of length $\lambda$}}$
\end{center}
such that each $v_i$ is a vertex of $K_n$, each $e_j$ is an edge of
$K_n$, and the vertices connected by $e_i$ are $v_i$ and $v_{i+1}$.

Let $A$ be the adjacency matrix of $K_n$, such that $A$ is an $n$-square binary
matrix in which each entry is either zero or one, i.e., every $(i,j)$-entry in
$A$ is equal to the number of edges incident to $v_i$ and $v_j$. Moreover, $A$
is symmetric and circulant \cite{Barry2007}. It has always zeros on the leading
diagonal and ones off the leading diagonal. For example, the adjacency matrix of
a complete graph $K_4$ is always equal to:
\begin{eqnarray*}
A = \begin{bmatrix}
0 & 1 & 1 & 1\\
1 & 0 & 1 & 1\\
1 & 1 & 0 & 1\\
1 & 1 & 1 & 0\\
\end{bmatrix}
\end{eqnarray*}
The total number of possible walks of length $\lambda$ from vertex $v_i$ to
vertex $v_j$ is the $(i,j)$-entry of $A^\lambda$, i.e., the matrix product,
denoted by ($\cdot$), of $\lambda$ copies of $A$ \cite{Stanley2000}. Following
the above example, the number of walks of length $2$ between any two distinct
vertices can be obtained directly from $A^2$, such that
\begin{eqnarray*}
A^2 = A \cdot A =
\begin{bmatrix}
(n-1) & (n-2) & (n-2) & (n-2)\\
(n-2) & (n-1) & (n-2) & (n-2)\\
(n-2) & (n-2) & (n-1) & (n-2)\\
(n-2) & (n-2) & (n-2) & (n-1)\\
\end{bmatrix}
\end{eqnarray*}
which leads to
\begin{eqnarray*}
A^2 = A \cdot A =
\begin{bmatrix}
3 & 2 & 2 & 2\\
2 & 3 & 2 & 2\\
2 & 2 & 3 & 2\\
2 & 2 & 2 & 3\\
\end{bmatrix}
\end{eqnarray*}
Note that any $(i,j)$-entry of $A^2$ (where $i \neq j$) gives the same
number of walks of length $2$ from any two distinct vertex $v_i$ to
vertex $v_j$. The total number of walks of length $2$ between any two
distinct vertices can, thus, be obtained by consecutively adding the
values of every $(i,j)$-entry off the leading diagonal of matrix
$A^2$. In the above example, it suffices to sum $(4(4-1))$ times
(i.e., twice the number of edges in $K_4$) the value $2$ that any
$(i,j)$-entry (where $i \neq j$) has in $A^2$. This amounts to having
exactly $24$ possible walks on any $K_4$ graph.

Therefore, the problem of finding the number of walks of length $\lambda$
between any two distinct vertices of a $K_n$ graph reduces to finding the
$(i,j)$-entry of $A^\lambda$, where $i \neq j$. Indeed, let $a^\lambda_{i,j}$ be
the $(i,j)$-entry of $A^\lambda$. Then, the recurrence relation between the
original adjacency matrix $A$, and the matrix product of up to $\lambda-1$
copies of $A$, i.e.,
\begin{equation}
\label{eq:recursive-relation1}
  A^\lambda = A^{\lambda -1} \cdot A
\end{equation}
with initial conditions:\\
\begin{tabular}{cc}
  \begin{minipage}{4cm}
      \begin{equation*}
        a^2_{i,j} = \left\{
        \begin{array}{ll}
          (n-2) & \text{if}~~i \neq j\\
          (n-1) & \text{if}~~i=j
        \end{array},
        \right.
      \end{equation*}
    \end{minipage}
    &
    \begin{minipage}{4cm}
      \begin{equation*}
        a^1_{i,j} = \left\{
        \begin{array}{ll}
          1 & \text{if}~~i \neq j\\
          0 & \text{if}~~i=j
        \end{array}
        \right.
      \end{equation*}
    \end{minipage}
  \end{tabular}\\

\noindent is sufficient to solve the problem. Notice, moreover, that the result does not
depend on any precise value of either $i$ or $j$. Indeed, it is proved in
\cite{Stanley2000} that there is a constant relationship between the
$(i,j)$-entries off the leading diagonal of $A^\lambda$ and the $(i,j)$-entries
on the leading diagonal of $A^\lambda$. More precisely, let $t^\lambda$ be any
$(i,j)$-entry off the leading diagonal of $A^\lambda$ (i.e., $t^\lambda =
a^\lambda_{i,j}$ such that $i \neq j$). Let $d^\lambda$ be any $(i,i)$-entry on
the leading diagonal of $A^\lambda$ (i.e., $t^\lambda = a^\lambda_{i,i}$). Then,
if we subtract $t^\lambda$ from $d^\lambda$, the results is always equal to
$(-1)^\lambda$. In other words, if we express $A^\lambda$ as follows:
\begin{equation*}
  A^\lambda = [ a^\lambda_{i,j} ] = \left\{
    \begin{array}{ll}
      t^\lambda & \text{if}~~i \neq j\\
      d^\lambda & \text{if}~~i=j
    \end{array}
    \right.
\end{equation*}
then $t^\lambda = d^\lambda + (-1)^\lambda$. We can now use the recurrence
relation shown in Equation (\ref{eq:recursive-relation1}) to derive the
following two results:
\begin{eqnarray}
    t^\lambda &=& (n-2) t^{\lambda -1} + d^{\lambda-1}
\label{eq:recursive-relation2}\\
    d^\lambda &=& (n-1) t^{\lambda -1} \label{eq:recursive-relation3}
\end{eqnarray}
\medskip
with the initial conditions $t^1 = 1$ and $d^1 = 0$.

Cumbersome, but elementary, transformations shown in both
\cite{Barry2007} and \cite{Stanley2000} lead us to unfold the two
recurrence relations in both Equation (\ref{eq:recursive-relation2})
and (\ref{eq:recursive-relation3}) to the following two self-contained
expressions:
\begin{eqnarray}
    t^\lambda & = & \frac{(n-1)^{\lambda} - (-1)^{\lambda}}{n}
\label{eq:non-closed-walks} \\
    d^\lambda & = & \frac{(n-1)^{\lambda} + (n-1)(-1)^{\lambda}}{n}
\label{eq:closed-walks}
\end{eqnarray}

To conclude, we can now use Equations (\ref{eq:non-closed-walks}) and
(\ref{eq:closed-walks}) to express the total number of closed and
non-closed walks in the complete graph $K_n$ by simply adding to them
twice the number of edges in the graph, i.e., $n(n-1)$. From
Equation (\ref{eq:non-closed-walks}) we have now the value of any
$(i,j)$-entry in $A^\lambda$ such that $i \neq j$. As we did
previously in the example of the complete graph $K_4$, the total
number of walks of length $\lambda$ between any two distinct vertices
can be obtained by consecutively adding $n(n-1)$ times the values of
any of the $(i,j)$-entries off the leading diagonal of matrix
$A^\lambda$. This amounts to having exactly $n(n-1) \cdot t^\lambda$
which simplifying leads to:
\begin{eqnarray}
  ((n-1)((n-1)^{\lambda} - (-1)^{\lambda}))
\label{eq:final-expression}
\end{eqnarray}
possible walks of length $\lambda$ on any $K_n$ graph.

\end{document}